\documentclass[a4paper,12pt,reqno]{amsart}
\usepackage[T1]{fontenc}
\usepackage[utf8]{inputenc}
\usepackage[libertinus,vvarbb]{newtx}
\usepackage[margin=1in]{geometry}
\usepackage{enumitem}
\usepackage{xparse}
\usepackage{graphicx}
\usepackage{subfig}

\usepackage[mathscr]{eucal}
\usepackage{rsfso}
\usepackage{xcolor}
\usepackage[bookmarksdepth=2]{hyperref}
\usepackage{multirow}

\numberwithin{equation}{section}

\newtheorem{theorem}{Theorem}[section]
\newtheorem{lemma}[theorem]{Lemma}

\newtheorem{proposition}[theorem]{Proposition}

\theoremstyle{definition}

\newtheorem{assumption}[theorem]{Assumption}

\newcommand{\R}{\mathbb{R}}

\newcommand{\N}{\mathbb{N}}
\newcommand{\E}{\mathbb{E}}
\newcommand{\Q}{\mathbb{Q}}
\renewcommand{\P}{\mathbb{P}}

\newcommand{\ind}{\mathbb{1}}
\newcommand{\stf}{S_t^\mathcal{F}}
\newcommand{\stc}{S_t^\mathcal{C}}

\renewcommand{\leq}{\leqslant}
\renewcommand{\geq}{\geqslant}

\NewDocumentCommand{\formula}{ssom}{%
 \IfBooleanTF{#1}{%
  \IfBooleanTF{#2}{%
   \IfValueTF{#3}%
    {\begin{align}\label{#3}\begin{gathered}#4\end{gathered}\end{align}}%
    {\begin{gather}#4\end{gather}}%
  }{%
   \IfValueTF{#3}%
    {\begin{align}\label{#3}\begin{aligned}#4\end{aligned}\end{align}}%
    {\begin{gather*}#4\end{gather*}}%
  }%
 }{%
  \IfValueTF{#3}%
   {\begin{align}\label{#3}#4\end{align}}%
   {\begin{align*}#4\end{align*}}%
 }%
}

\begin{document}

\title[Design and valuation of multi-region CoCoCat bonds]{Design and valuation of multi-region CoCoCat bonds}
\author{Jacek Wszoła, Krzysztof Burnecki, Marek Teuerle, Martyna Zdeb}
\thanks{Work supported by NCN Grant No. 2022/47/B/HS4/02139}
\address{\normalfont Jacek Wszoła, Krzysztof Burnecki, Marek Teuerle, Martyna Zdeb \\ Faculty of Pure and Applied Mathematics \\ Wrocław University of Science and Technology \\ Wyspiańskiego 27 \\ 50-370 Wrocław, Poland}
\email{krzysztof.burnecki@pwr.edu.pl}
\date{\today}

\begin{abstract}
This paper introduces a novel multidimensional insurance-linked instrument: a contingent convertible bond (CoCoCat bond) whose conversion trigger is activated by predefined natural catastrophes across multiple geographical regions. We develop such a model explicitly accounting for the complex dependencies between regional catastrophe losses. Specifically, we explore scenarios ranging from complete independence to proportional loss dependencies, both with fixed and random loss amounts. Utilizing change-of-measure techniques, we derive risk-neutral pricing formulas tailored to these diverse dependence structures. By fitting our model to real-world natural catastrophe data from Property Claim Services, we demonstrate the significant impact of inter-regional dependencies on the CoCoCat bond's pricing, highlighting the importance of multidimensional risk assessment for this innovative financial instrument.
\end{abstract}

\maketitle

\section{Introduction}

Insurance-linked securities (ILS) represent a class of financial instruments designed to transfer specific insurance-related risks, predominantly those associated with natural catastrophes but increasingly encompassing other areas such as cyber, mortality, and longevity risks, from sponsoring entities to capital market investors \cite{Canabarro2000,Barrieu2010, Braun2022Cyber,Lane2024Market}. Sponsors typically include insurance and reinsurance companies seeking alternative risk capital or capacity, but also governments and corporations seeking to manage exposure to extreme events \cite{hc,zimbidis2023}. The fundamental mechanism involves the sponsor paying a premium or coupon to investors in exchange for financial protection against predefined trigger events. This protection is facilitated through a Special Purpose Vehicle (SPV), an entity created specifically for the transaction, which issues the securities (often bonds) to investors. The proceeds from the issuance are held in a collateral account, usually invested in highly secure, liquid assets. If no triggering event occurs during the security's term, investors receive their principal back along with the agreed-upon coupon payments, which typically offer a premium over risk-free rates. However, if a predefined trigger event (e.g., a hurricane of a certain intensity hitting a specific region, industry losses exceeding a threshold, or the sponsor's own losses surpassing a level) occurs, the principal, and sometimes accrued interest, is partially or fully forgiven and used to cover the sponsor's claims arising from the event \cite{Barrieu2010}.

This structure offers potential benefits to both parties. For sponsors, ILS provide access to the vast capacity of the capital markets, which significantly exceeds that of the traditional reinsurance market, allowing for the transfer of large, peak risks that might otherwise be difficult or expensive to reinsure. It diversifies their sources of risk capital, potentially reduces counterparty credit risk (as ILS are typically fully collateralized), and can offer multi-year coverage stability \cite{Canabarro2000}. For investors, ILS offer the potential for attractive yields and, crucially, returns that are largely uncorrelated with traditional financial assets like stocks and bonds, since the underlying risks (e.g., natural disasters) are generally independent of economic cycles. This low correlation makes ILS a valuable tool for portfolio diversification \cite{carayannopoulos2015diversification,Chakrabatu2017}.

While early ILS transactions often focused on single, well-understood peak perils in specific regions (e.g., US hurricane risk), the market has evolved to include more complex structures covering multiple perils and/or multiple geographic regions. Multi-peril ILS bundle coverage for different types of risks (e.g., earthquake and windstorm) within a single instrument, while multi-region ILS cover risks across different geographical areas (e.g., US and Europe, or multiple countries within a region). Examples include bonds covering US wind and earthquake jointly, North American multiperil risks, or combinations like European wind and Turkish earthquake \cite{SwissRe2024Insights, Lane2024Market}.
In \cite{burnecki2024multiPeril} a framework for pricing insurance-linked securities that are associated with multiple natural catastrophe risks was developed. As a representative case, a multi-peril catastrophe bond was designed that can be linked either to industry loss indices or to the actual losses experienced by an insurer.

The literature proposes several theoretical frameworks to address the challenges of pricing ILS in incomplete markets.  Arbitrage-Free Pricing (Risk-Neutral Framework) adapts standard asset pricing theory. One method assumes investors are risk-neutral specifically towards the non-systematic jump risk associated with catastrophe events, implying the risk-neutral probability measure for these events coincides with the real-world physical probability measure \cite{taylor,cox2000catastrophe}. Alternatively, pricing can be performed under an equivalent martingale measure derived using techniques such as the Esscher transform or the Wang transform, which explicitly incorporate a market price of the risk parameter to adjust probabilities \cite{BAUER2010139,wang2004cat}. Another approach is contingent claim analysis (CCA). This framework values ILS, particularly catastrophe bonds, as derivative securities, similar to options written on an underlying measure of catastrophe loss or event occurrence \cite{mm,nr,vaugirard2003pricing}. It is also worth mentioning utility-based frameworks. These models analyse the decision to issue or invest in ILS from the perspective of maximising expected utility \cite{LIU2021793,petra2021}. 

Complementing theoretical models, empirical and statistical approaches aim to explain observed ILS prices or forecast future prices based on data. Early empirical work focused on identifying the key determinants of ILS spreads (the premium over the risk-free rate) at issuance using linear or log-linear regression models \cite{lane2003}. Common explanatory variables include measures of risk (expected loss (EL), probability of first loss (PFL)), bond characteristics (term, size, trigger type, peril, region, sponsor identity/quality, credit rating), and market conditions (reinsurance market cycle indicators like Rate-on-Line indices, general credit market spreads like high-yield bond spreads) \cite{braun2,LaneBeckwith2021}. More recent research uses machine learning (ML) techniques, such as gradient boost machines (e.g., XGBoost), random forests and geometric deep learning approach, to model ILS prices \cite{Lane2018,GotGurWit2020, Chen2024ML,domfeh2025catnetgeometricdeeplearning}. ML models excel at capturing complex, non-linear dependencies and interaction effects among a large number of potential predictors without requiring pre-specified functional forms. 

A crucial input for many pricing models is the statistical distribution of potential losses or the probability of trigger events. Research focuses on fitting appropriate probability distributions (e.g., generalised extreme value (GEV), generalized Pareto distribution (GPD), Burr, lognormal) to historical catastrophe loss data, often using extreme value theory \cite{mm,Manathunga2023, Zimbidis2007}. For multi-peril or potentially dependent risks, copula functions are used to model the dependence structure between marginal loss distributions \cite{wei2022copulaEVT,tang2023multiEvent}. These statistical analyses often need to account for data limitations like left-truncation (only losses above a certain threshold being recorded) \cite{giuricich,burnecki2024multiPeril}.

In ILS, event arrivals are typically modelled by counting processes: homogeneous (HPP) or non-homogeneous Poisson processes (NHPPs) for single-peril frequency (with seasonality captured via time-varying intensities), as in actuarial treatments of seasonal CAT bonds and flood-hedging CAT bonds that use NHPPs with trend and seasonality \cite{mm,burjanwer11}. To allow intensity to jump and decay after shocks, and to accommodate over-dispersion relative to Poisson, many pricing frameworks use Cox (doubly stochastic Poisson) models, often with shot-noise intensities; these are standard in catastrophe derivatives and CAT-bond pricing \cite{dj}. When exponential inter-arrival times are too restrictive, renewal processes provide a flexible alternative for inter-event timing and have been used directly in CAT-bond valuation \cite{burgiu17risks,Safarveisi05052025}. 
 For perils with clustering—most notably earthquakes—self-exciting point processes such as Hawkes/ETAS better capture aftershock cascades and materially affect risk-based CAT-bond metrics \cite{mistry2024,Louloudis2024Robust}. 
 Finally, portfolio-level dependence across regions/perils is often handled with common-shock Poisson frameworks that correlate event counts while preserving marginal structures, a device widely adopted in insurance risk modelling and applicable to multi-peril ILS \cite{Lindskog_McNeil_2003}.

The potential impact of climate change adds another layer of complexity \cite{Morana2019}. Climate change may alter the frequency, intensity and geographical distribution of weather-related risks such as hurricanes, floods, and wildfires, which may make historical data less relevant for predicting future risk \cite{sigma2020}. 

Recently in the banking industry, ''contingent capita'' instruments, such as contingent convertible (CoCo) bonds, have gained the support of various academics, practitioners, economists, regulators and banks as a potential avenue to reduce the need for bailouts of institutions that are classified as `too-big-to-fail' \cite{ flannery2016stabilizing}. Contingent capital instruments are a type of debt instrument with a loss-absorbing mechanism: that is, they are automatically converted into common equity or written down when a pre-specified trigger event occurs \cite{flannery2016stabilizing}.

Given the success of contingent capital instruments, such as CoCo bonds, in \cite{bgp} the authors specified a special type of CoCo bond for insurers and reinsurers. 
Such a CoCoCat bond can be seen as a special type of CoCo bond. CoCo bonds are characterised by two important features, namely the conversion trigger and the conversion mechanism.


In this paper, we consider a CoCoCat bond which is a contingent convertible bond that has a trigger linked to the occurrence of a single or sequence of predefined natural catastrophes in different regions, and a conversion mechanism whereby the bond either (i)
converts into common equity of the issuer (therefore increasing the size of common equity
in issue), at a predefined conversion rate as specified in the bond covenant, or (ii) is written
down (both principal and coupons) by a fixed percentage which is specified in the covenant. 
We construct a model that explicitly captures the intricate dependence patterns among regional catastrophe losses. Our analysis covers a range of scenarios, from fully independent events to proportionally dependent losses, considering both fixed and random loss magnitudes. By applying change-of-measure methods, we derive risk-neutral pricing formulas that reflect these various dependency structures.


The structure of the paper is as follows. Section \ref{sec:mode} introduces the concept of CoCoCat models applied across multiple regions, examining three distinct scenarios: fully independent loss processes, independent loss amounts, and proportional loss amounts under both fixed and random settings. In Section \ref{sec:price}, we apply change-of-measure techniques to develop a valuation framework and derive risk-neutral pricing formulas tailored to these scenarios. The results are categorized into three groups based on the underlying assumptions about the loss mechanisms. Section \ref{sec:proof} presents detailed proofs of the key results from Section \ref{sec:price}. Section \ref{sec:gen} discusses potential extensions of the model to accommodate claims from three or more regions. In Section \ref{sec:gen}, we perform numerical analyses using the pricing formulas from Section \ref{sec:price} to explore how CoCoCat bond prices respond to changes in model parameters. The catastrophe loss models are calibrated using data from Property Claims Services. The paper concludes in Section \ref{sec:con} with a summary of the main findings.

\section{Considered models}
\label{sec:mode}

We lay here the groundwork for the valuation of multi-peril and multi-region CoCoCat bonds by detailing the mathematical frameworks employed. 
The section begins by establishing fundamental assumptions.

\subsection{Assumptions} 

Loosely following \cite{jaimungal2006catastrophe}, we commence under the real-world probability measure. Under any probability measure, the price of the CAT bond depends on two emerging phenomena: financial market-related risk and catastrophe-related risk. Since the catastrophe-risk variables will give rise to jumps, we need to work in an incomplete markets setting and, moreover, note that complicated changes of measure could arise. To avoid this, we make the following assumption in line with much of the previous literature on pricing catastrophe-linked financial instruments. Evidence in support of this assumption has been found by \cite{hoyt} and \cite{cummins2009convergence}, but is disputed by \cite{carayannopoulos2015diversification} and \cite{hagendorff2015impact}, see also the discussion in \cite{bgp}.

\begin{assumption}
\label{ass_1}
    Catastrophe risk variables and financial market risk variables are independent in the real world.
\end{assumption}

We now continue our analysis by considering two financial-market-related processes. For the purposes of CAT bond pricing, we assume that on this space exists a riskless bank account based on an interest rate process, which we shall adopt as the num\'eraire asset later.

We also follow the incomplete market framework of \cite{merton}. Such a framework has been used extensively in the literature when valuing derivatives with payoffs linked to natural catastrophes; see, for example, \cite{bakshi2002average}, \cite{ly}, \cite{vaugirard2003pricing}, \cite{jaimungal2006catastrophe}, \cite{ly2}, \cite{mm}, \cite{nr}, and \cite{chang2017integrated}. On the basis of its pervasiveness in the literature to date, the following assumption is made in our work.

\begin{assumption}
\label{ass_2}
    Investors are risk-neutral towards the jump risk posed by the natural catastrophe risk variables.
\end{assumption}

However, it must be borne in mind that recent empirical catastrophe bond pricing literature has shown that catastrophe bonds do not have a zero risk premium, see, for example, \cite{ppd}, \cite{braun2} and \cite{gurtler2016impact}. This effect could also extend to other catastrophe-linked ILS instruments. Against this backdrop, it is possible to infer that pricing models based on the zero risk-premium assumption may give rise to values lower than those given by pricing models which assume a non-zero risk premium. In consequence, the usage of these pricing models may require additional margins added to the calculated value, or margins added to the parameters of the distributions associated with the jump process, all at the discretion of the issuer. 

Despite this, we remain true to Assumption \ref{ass_2} in our work for two reasons. First, it aligns with actuarial pricing techniques, which, as noted by \cite{braun}, dominate in practice. Second, and more importantly, it can be adjusted to the range of underlying state variables in the model that are not investment assets and, therefore, not tradable. in the model. Hence, we can use real-world data to price, and this is useful given the scarcity of (and difficulty of obtaining) pricing data for many catastrophe-linked ILS.


In this paper, we consider three specialised cases of multi-region loss processes, which are crucial for subsequent pricing formulae. These models provide the foundation for understanding how catastrophic events in multiple regions and their financial implications are structured and analysed within the paper.  
First, let us introduce a general 2D pricing framework.

\subsection{General 2D pricing framework}
The general construction builds on the 1D framework introduced in \cite{bgp}. The 2D framework describes both the financial market risk and the catastrophe risk variables. The former follow the classical Black-Scholes dynamics with appropriately selected stochastic interest rate process, while in the latter the two-dimensional aggregate loss process which can be modelled by various compound Poisson processes.

First, we establish some notation on the parameters of the CoCoCat bond, which covers the risk in two regions. Let:
\begin{itemize}
    \item $T>0$ be the term of the CoCoCat bond;
    \item $Z>0$ be the principal invested;
    \item $0<t_1< t_2< \ldots< t_N = T$ be the coupon paying dates with constant yearly time period between the dates: $\Delta = t_{i}-t_{i-1}$ for $i = 2,3,\ldots,N$;
    \item $c \geq 0$ be the constant spread (i.e. the catastrophe risk premium);
    \item $\zeta \in [0,1]$ be the conversion fraction;
    \item $D_1,D_2 > 0$ be the threshold levels for the trigger corresponding to two selected regions;
    \item $K_P >0$ be the conversion price.
\end{itemize}

Moreover, let $\{S_t: t\geq0\}$ denote the share price process of the issuing firm. It is natural to express it
in terms of financial and catastrophic risk variables. 
Following the idea of Assumption \ref{ass_1} we assume that it can be decomposed as
\formula[eq:ststfstc]{
S_t = \stf \stc,
}
where $\{S_t^\mathcal{F}: t \geq 0\}$ and $\{S_t^\mathcal{C}: t\geq 0\}$ are two independent processes driven by financial market risk and catastrophic risk, respectively. Additionally, we set $S_0$ to be the price of the issuing firm at time $t=0$.

The part of the model corresponding to the financial world is very similar to the 1D case in \cite{bgp}. We use the standard Black--Scholes dynamics and for the structure of interest rate $\{r_t: t \geq 0\}$ we select Longstaff's model, namely under the real-word measure $\mathbb{P}$
\formula[eq:drt]{
dr_t &= \hat{\theta}_r(\hat{m}_r-\sqrt{r_t})\,dt +\hat{\sigma}_r \sqrt{r_t} d\hat{W}_t^1, \\
d \stf &= \mu_S \stf dt + \sigma_S \stf d\hat{W}_t^2.
}
The processes $\hat{W}_t^1$ and $\hat{W}_t^2$ are Brownian motions, which satisfy
\formula[eq:rho]{
d \langle \hat{W}_t^1, \hat{W}_t^2 \rangle = \rho \, dt
}
for some $\rho \in [-1,1]$. 
These assumptions imply that
\formula[stf-formula]{
\stf = S_0 \exp\left( \sigma_S \hat{W}_t^2 + \int_0^t (r_u - \tfrac{1}{2} \sigma_S^2 ) du \right).
}

Note that instead of Longstaff's model, one can choose any interest rate model invariant to the Girsanov transformation with a constant kernel, such as Vasicek's model or Hull-White's model. Moreover, all of these three models admit a closed-form solution to the price of a zero-coupon bond paying 1 at maturity. 

Now, we characterise the part of the model corresponding to the catastrophe-risk variables. We define
\formula[stc-formula]{
\stc = \exp \left( -\alpha L_t^1 - \beta L_t^2 + \alpha \kappa_1  \int_0^t \lambda_u^1 \, du + \beta \kappa_2 \int_0^t \lambda_u^2 \, du \right),
}
where $\alpha, \beta, \kappa_1, \kappa_2 >0$ and $L_t^1$, $L_t^2$ denote the aggregate loss processes given by
\formula{
& L_t^1 = \sum_{k=1}^{N_t^1} X_k^1, & L_t^2 = \sum_{k=1}^{N_t^2} X_k^2.
}
Here, $N_t^1$ and $N_t^2$ are non-homogeneous Poisson processes with cumulative deterministic intensities
\formula{
& \Lambda_t^1 = \int_0^t \lambda_u^1 \, du, & \Lambda_t^2 = \int_0^t \lambda_u^2 \, du,
}
respectively, for some non-negative intensity functions $\lambda_u^1$ and $\lambda_u^2$.

The loss amounts described by non-negative continuous random variables $X_1^1, X_2^1, X_3^1 \ldots$ and $X_1^2, X_2^2, X_3^2, \ldots$ are pairwise independent and identically distributed with distribution functions $F_X^1$, $F_X^2$ and densities $f_X^1$, $f_X^2$, respectively. By pairwise independent, we mean here that for any $i=1,2$ and $k,l \in \N$, $k \neq l$ random variables $X_k^i$ and $X_j^i$ are independent. In particular, note that this definition allows for the dependency of variables $X_k^1$ and $X_l^2$ for some $k,l \in \N$. This is a natural assumption regarding the fact that these variables are identified with loss amounts in two different regions.

The aggregate loss processes $L_t^1$ and $L_t^2$ correspond to the behaviour of CoCoCat bond's $i$th underlying index. The constants $\alpha, \beta$ represent the effect of catastrophic losses on the logarithm of the share price. The coefficients $\kappa_1$ and $\kappa_2$ correspond to Poisson processes $N_t^1$ and $N_t^2$, respectively. Specifically, if $N_t^1$ and $N_t^2$ are the same process, which we will denote by $N_t$, then $\kappa_1 = \kappa_2$ and we denote this common value as $\kappa$.


We consider three special cases of the aggregate loss process and, for these cases, we will provide appropriate pricing formulae. We now give precise mathematical assumptions for these cases and provide their interpretation.

\subsection{Independent loss processes (ILP)}\label{ilp-definition}

Suppose that we observe losses from two different regions with two different frequencies. Mathematically speaking, we assume that the variables~$N_t^1, N_t^2$, $X_k^1, X_k^2$ are (pairwise) independent for any $k \in \N$. Clearly, the aggregated loss processes $L_t^1$ and $L_t^2$ are also independent. Therefore, we can rewrite \eqref{eq:ststfstc} as
\formula[ilp-stc-product]{
\stc = S_t^{\mathcal{C},1} S_t^{\mathcal{C},2},
}
where
\formula{
& S_t^{\mathcal{C},1} = \exp\left(-\alpha \sum_{k=1}^{N_t^1} X_k^1 + \alpha \kappa_1 \int_0^t \lambda_u^1 \, du\right), &   S_t^{\mathcal{C},2} = \exp\left(-\beta \sum_{k=1}^{N_t^2} X_k^2 + \beta \kappa_2 \int_0^t \lambda_u^2 \, du\right).
}
Note that $\smash{S_t^{\mathcal{C},1}}$ and $\smash{S_t^{\mathcal{C},2}}$ are independent single-region processes for the 1D model (see equation (7) in \cite{bgp}). Formula \eqref{ilp-stc-product} provides a very convenient way to represent the catastrophic share price process, allowing the use of 1D methods.

Another way of thinking about the ILP assumption is to treat the two distinct regions as one and simply reduce the problem to the 1D case. One can define the merged aggregate loss process: 
\formula[ilp-lt-merged]{
L_t = \frac{\alpha}{\alpha+ \beta} L_t^1 + \frac{\beta}{\alpha+\beta} L_t^2.
}

It is easy to see that $L_t$ is a compound Poisson process itself. Its frequency component reads $N_t^1 + N_t^2$ with time-dependent cumulative intensity $\Lambda_t = \Lambda_t^1 + \Lambda_t^2$. Moreover, $L_t$ has random loss amounts $X_1, X_2, X_3, \ldots$ which are i.i.d. random variables following the mixture of distributions $F_X^1$ and $F_X^2$. However, it should be noted that, in general, the distribution of $X_k$ may be time-dependent, unlike the distributions of $X_k^1$ and $X_k^2$. This problem does not arise when $N_t^1$ and $N_t^2$ are homogeneous Poisson processes.

In this notation, equation \eqref{stc-formula} reads
\formula{
\stc &= \exp \left( -(\alpha+\beta) L_t + \alpha \kappa_1 \int_0^t \lambda_s^1 \, ds + \beta \kappa_2 \int_0^t \lambda_s^2 \, ds \right).}

Although this way of seeing the ILP assumption can seem very natural, it will not be very helpful in pricing, mainly because of non-matching integral components. However, it can be considered an inspiration to use this trick for the two remaining cases.


\subsection{Independent loss amounts (ILA)}\label{ila-definition}

Suppose that losses from two distinct regions occur at the same time. It means that the loss amount variables $X_k^1$ and $X_k^2$ are independent in the sense that for any $k,l \in \N$, $i,j = 1,2$ such that $(i,k) \neq (j,l)$ variables $X_k^i$ and $X_l^j$ are independent. Moreover, the processes $N_t^1$ and $N_t^2$ are the same, so we denote both as $N_t$. The cumulative intensity of $N_t$ is denoted as $\Lambda_t$.

We can again define the aggregate loss process merged in the same vein as in \eqref{ilp-lt-merged}, but since for ILA loss frequency is the same, it simplifies to
\formula[ila-lt-merged]{
L_t = \frac{\alpha}{\alpha+\beta} L_t^1 + \frac{\beta}{\alpha + \beta} L_t^2 = \sum_{k=1}^{N_t} \left( \frac{\alpha}{\alpha + \beta} X_k^1 + \frac{\beta}{\alpha + \beta} X_k^2 \right).
}

One can easily identify the distribution of the summands $X_k$, as they are the sums of two independent random variables.

Since under the ILA assumption the losses occur simultaneously, we have $\kappa_1 = \kappa_2=\kappa$. Equation $\eqref{stc-formula}$ can be rewritten again in terms of the merged aggregate loss process $L_t$ as follows:
\formula{
\stc &= \exp \left( -(\alpha+\beta)L_t + \kappa(\alpha+\beta) \int_0^t \lambda_s \, ds \right).
}

Thus, it can be seen that in this approach, the process $\stc$ can be reduced to its single-region analogue.


\subsection{Proportional loss amounts (PLA)}\label{pla-definition}
Now, let us consider a case where not only the losses occur at the same time but the losses $X_k$ themselves are split proportionally among the regions. Similarly to the ILA assumption, we denote the common counting Poisson process by $N_t$ and its cumulative intensity by $\Lambda_t$. We consider two types of proportional splits of losses:
\begin{enumerate}[label = (\alph*)]
    \item\label{a} Constant proportional loss amounts (cPLA): for fixed and deterministic proportion coefficient $p \in (0,1)$, the loss amounts satisfy:
    \formula{
    & X_k^1 = p X_k, & X_k^2 = (1-p) X_k, &
    }
    where $X_k$ is a sequence of i.i.d. random variables with distribution function $F_X$ and density $f_X$.  
    \item\label{b} Random proportional loss amounts (rPLA): for random proportion coefficient $P \in (0,1)$, loss amounts satisfy:
    \formula{
    & X_k^1 = P X_k, & X_k^2 = (1-P) X_k. &
    }
    We additionally assume that $P$ is independent of $X_k$ and $P$ and has a distribution function $F_P$. 
\end{enumerate}

Note that for the rPLA case we do not assume that $P$ is a continuous random variable. This allows us to recover the cPLA case simply by defining $P = p$ with probability 1 for $p \in (0,1)$. For that reason, throughout the rest of the paper we mainly focus on the more general case (random $P$).

For PLA, it is more convenient to slightly modify the previous definition of the merged aggregate loss process and define it as
\formula[pla-lt-merged]{
\alpha L_t^1 + \beta L_t^2  = (\alpha P + \beta(1-P)) \sum_{k=1}^{N_t} X_k = Q L_t,
}
where $Q = \alpha P + \beta (1-P)$ and $L_t = \sum_{k=1}^{N_t} X_k$. Here, $L_t$ refers to the standard aggregate loss process. Therefore, the catastrophic share price process can be written as
\formula{
\stc = \exp\left( - Q L_t + \kappa (\alpha + \beta) \int_0^t \lambda_u \, du \right).
}

Note that for certain special cases, some simplifications are possible. For $\alpha = \beta$ we have $Q=\alpha$ and hence $\stc$ reduces to the single-region case, with loss amounts given by $\alpha X_k$ and effect of losses ($\alpha+\beta$) equal to $2\alpha$. Another example is when $P=\frac{1}{2}$ with probability 1 (or $p=1/2$). Then we obtain another single-region process with the losses $(\alpha+\beta)X_k/2$ and effect of losses $\alpha+\beta$. 

Another possible assumption to study would be a further generalisation of the rPLA case and considering floating proportional coefficients $P_k$, for which the losses would be $X_k^1 = P_k X_k$ and $X_k^2 = (1-P_k) X_k$. Despite its realistic nature, this assumption poses many difficulties, such as finding the distribution of the trigger time. However, it opens the door to further development of multi-region framework.


\section{Risk-neutral pricing: analytic formulae}
\label{sec:price}

Under the risk-neutral measure $\mathbb{Q}$ the catastrophe-risk and financial market risk variables are captured by the following system
of equations.

\begin{proposition}
The multi-region model is defined by the following system of equations:
\formula[]{
\label{rt} dr_t &= \theta_r(m_r-\sqrt{r_t})\,dt +\sigma_r \sqrt{r_t} dW_t^1, \\
\label{st} S_t &= \stf \stc, \\
\label{stc} \stc &= \exp \left( -\alpha L_t^1 - \beta L_t^2 + \alpha \kappa_1 \int_0^t \lambda_u^1 \, du + \beta \kappa_2 \int_0^t \lambda_u^2 \, du \right), \\
\label{stf} d \stf &= \mu_S \stf dt + \sigma_S \stf dW_t^2, \\
\label{w1w2} d \langle W_t^1, W_t^2 \rangle &= \rho dt, \\
\label{lts} (L_t^1, L_t^2) &= \left(\sum_{k=1}^{N_t^1} X_k^1, \, \sum_{k=1}^{N_t^2} X_k^2,\right)
}
where $\theta_r$ and $m_r$ are the risk-neutral parameters for the interest rate process and
 $\tilde{W}^1_t$ and $\tilde{W}^2_t$ are two Brownian motions under the measure $\mathbb{Q}_F$ as specified in \cite{bgp}.
\end{proposition}
We now go into the details of pricing.

\subsection{CoCoCat bond's mechanism. General pricing formula}

Before we present the main results of the paper, we specify some additional notation. Let $R(t, t_{i-1}, t_{i})$ be the forward risk-free (like LIBOR) rate at time $t$ for the interval $[t_{i-1}, t_i]$. Since $\Delta = t_i - t_{i-1}$ is constant, the risk-free process $R(t,t,t+\Delta)$ at time $t$ is denoted as $\{R_t: t \geq 0\}$. We also define the discounted riskless bank account associated with the process $\{r_t: t \geq 0\}$ as
\formula[b0t]{
B(0,t) = \exp\left(-\int_0^t r_s \, ds\right).
}
Hence, risk-neutral price at time $t$ of a zero-coupon bond paying one unit at maturity $T$ ($T>t$) can be calculated as the conditional expectation given by $\mathbb{E}^\mathbb{Q}[B(0,T)| \mathcal{F}_t]$ for $t \in [0,T]$. However, in pricing, we will focus on $t=0$. Note that in this case, the above expectation is no longer conditional since $r_0$ is deterministic. We denote the price of a zero-coupon unit bond with maturity $T$ as
\formula[p-bond]{
P(r_0, T, \theta_r, m_r, \sigma_r) = \mathbb{E}^\mathbb{Q} B(0,T)
}
when the interest rate $r_t$ is given by the model \eqref{rt} with parameters $\theta_r, \, m_r, \, \sigma_r$ and initial value $r_0$. The analytical formulae for \eqref{p-bond} are well-known in the literature. These famous results are recalled below.

\begin{proposition}\label{p-formulas}
    If $r_t$ follows Longstaff's model, we have
    \formula{
    P(r_0, T, \theta_r, m_r, \sigma_r) = A(T) \exp( r_0 B(T) + \sqrt{r_0} C(T)),
    }
    where
    \formula{
    A(T) &= \left( \frac{2}{1+e^{\psi T}}\right)^{1/2} \exp\left( c_1 + c_2 T + \frac{c_3}{1+e^{\psi T}} \right), \\
    B(T) &= - \frac{\psi}{\sigma_r^2} + \frac{2\psi}{\sigma_r^2 ( 1+e^{\psi T})}, \\
    C(T) &= \frac{2 \theta_r (1-e^{\psi T/2})}{\sigma_r^2(1+ e^{\psi T})}
    }
    and
    \formula{
    & \psi = \sqrt{2} \sigma_r, & c_1 = \frac{\theta_r^2}{\psi \sigma_r^2}, & & c_2 = \frac{\psi}{4} - \frac{\theta_r^2}{\psi^2}, & & c_3 = - \frac{4 \theta_r^2}{\psi^3}. &
    }
\end{proposition}
    
Let us recall the general CoCoCat bond mechanism. With the principal $K$ invested in the bond, the coupons with rate $c+\operatorname{LIBOR}$ are paid at times $t_1 < t_2< \ldots < t_N=T$ until the maturity date $T$ or upon the occurrence of a trigger, whichever happens first. If the trigger does not occur before maturity, all the money is returned to the investor. Otherwise, immediately upon the time of trigger, the conversion mechanism is activated and $\zeta K$ is converted to common equity.

This construction is very similar to the 1D case presented in \cite{bgp}. The major difference is that in our case two catastrophic indices are considered and the trigger is activated when $L_1 \geq D_1$ or $L_2 \geq D_2$. In other words, we define the trigger time as
\formula{\tau = \min\{\tau_1, \tau_2\},}
where
\formula{
& \tau_1 = \inf\{t \geq 0: L_t^1 \geq D_1\}, & \tau_2 = \inf\{t \geq 0: L_t^2 \geq D_2\}
}
are the times the aggregate loss processes $L_t^i$ first exceed the values $D_i$ for $i=1,2$. This definition of trigger time, despite being very natural, constitutes the main difficulty in valuation multi-region CoCoCat bonds, even though the catastrophic stock price process can often be reduced to a single-region case, as shown in previous sections.

By the mechanism of the CoCoCat bond, the following general pricing formula holds.
\begin{lemma}\label{is-lemma}
    The issue-date risk-neutral price of a multi-region  CoCoCat bond is
    \formula{
    V_0 = \mathbb{E}^\mathbb{Q}[I_1 + I_2 + I_3],
    }
    where $\Q$ is the risk-neutral measure and
    \formula{
    & I_1 = \sum_{i=1}^N (R_{t_{i-1}}+c) \Delta Z \ind_{\tau > t_i} B(0, t_i), & &
    I_2 = \frac{\zeta Z}{K_P} S_\tau \ind_{\tau \leq T} B(0, \tau), &
    & I_3 = Z \ind_{\tau > T} B(0, T). &
    }
\end{lemma}

In this paper, we consider exponential conversion functions, that is $K_P = S_\tau^\nu$, where $\nu \in [0,1]$. In particular, if $\nu = 0$, then the conversion amount does not depend on the share price at the trigger moment, as it is constant. On the other hand, if $\nu=1$, the conversion amount is equal to the share price at the trigger moment and $I_2$ simplifies significantly.

Given the appropriate martingale measure $\mathbb{Q}$, the above expectations can be calculated by repeating the reasoning presented in \cite{bgp} (see Sections 4.1-4.3 therein). For random variables $I_1, I_3$, results are similar to those from the 1D case, that is
\formula*[ei1]{
    \E I_1 &= Z \Delta (R_0 + c) P(r_0, t_1, \theta_r, m_r, \sigma_r) \Q(\tau > t_1) \\
    &+ Z\sum_{i=2}^N \Q(\tau > t_i) (P(r_0, t_{i-1}, \theta_r, m_r, \sigma_r) + (1-c\Delta) P(r_0, t_i, \theta_r, m_r, \sigma_r))
}
and
\formula[ei3]{
    \E I_3 = Z P(r_0, T, \theta_r, m_r, \sigma_r) \Q(\tau > T).
}
The only difference in the above formulas with respect to the 1D case is the distribution of $\tau$ which we will study separately for each model. 

Evaluation of the remaining expectation, $\E^{\Q} I_2$, requires more advanced tools leading to more complicated analytical formulas but still easy to use. This will be exploited now in detail.

\subsection{Main results}

In this section we present the main results, namely the analytic formulae for risk-neutral price of multi-region CoCoCat bonds. We split the results into three groups, based on the assumptions on the loss process (ILP, ILA, PLA).

For any appropriate function $f: (0, \infty) \to \R$, we denote its Laplace transform by
\formula{
(\mathcal{L}f)(z) = \int_0^\infty e^{-xz} f(x) \, dx.
}
The $n$th convolution power of the function $f$ is denoted by $f^{n*}$ for $n=1,2,3\ldots$ For convenience, we assume that $f^{0*}$ is identically equal to one.

\begin{theorem}\label{thm-main}
    Risk-neutral price of CoCoCat bond is equal to
    \formula{
    \E^\Q I_1 + \E^\Q I_2 + \E^\Q I_3.
    }
    Here $\E^\Q I_1$ is given by \eqref{ei1}, $\E^\Q I_3$ is given by \eqref{ei3} and
    \begin{enumerate}[label = (\alph*)]
        \item for ILP,
        \formula{
           \E^\Q I_2 = \xi Z S_0^{1-\nu} \int_0^T \exp\left( -\tfrac{1}{2}\nu(1-\nu)^2 \sigma_S^2 t \right) \Phi(t) P(r_0, t, \bar{\theta}_r, \bar{m}_r, \bar{\sigma}_r) \, F_\tau^\nu(dt),
        }
        where
        \formula{
        \Phi(t) &= \exp \bigg( -\Lambda_t^1 (1-(\mathscr{L}f_X^1)(\alpha(1-\nu)) - \Lambda_t^2 (1-(\mathscr{L}f_X^2)(\beta(1-\nu)) \\
        &+ (1-\nu) \left( \Lambda_t^1(1-(\mathcal{L}f_X^1)(\alpha)) + \Lambda_t^2 (1-(\mathcal{L}f_X^2)(\beta))\right) \bigg)
        }
        and $F_\tau^\nu$ is described by Proposition \ref{ilp-tau} with parameters $\Lambda_t^{\nu,1}$, $\Lambda_t^{\nu,2}$, $F_X^{\nu,1}$, $F_X^{\nu,2}$ described by Proposition \ref{ilp-lt-distribution};
        
        \item for ILA, 
        \formula{
           \E^\Q I_2 = \xi Z S_0^{1-\nu} \int_0^T \exp\left( -\tfrac{1}{2}\nu(1-\nu)^2 \sigma_S^2 t \right) \Phi(t) P(r_0, t, \bar{\theta}_r, \bar{m}_r, \bar{\sigma}_r) \, F_\tau^\nu(dt),
        }
        where
        \formula{
        \Phi(t) &= \exp \bigg( -\Lambda_t(1- (\mathcal{L}f_X^1)(\alpha (1-\nu)) (\mathcal{L}f_X^2)(\beta (1-\nu))) \\
        &+ (1-\nu) \Lambda_t (1- (\mathscr{L}f_X^1)(\alpha) (\mathscr{L}f_X^2)(\beta)) \bigg)
        }
        and $F_\tau^\nu$ is described by Proposition \ref{ila-tau} with parameters $\Lambda_t^{\nu}$, $F_X^{\nu,1}$, $F_X^{\nu,2}$ described by Proposition \ref{ila-lt-distribution};
        
        \item for PLA,
        \formula{
           \E^\Q I_2 &= \xi Z S_0^{1-\nu} \int_0^1 \bigg( \int_0^T \exp\left( -\tfrac{1}{2}\nu(1-\nu)^2 \sigma_S^2 t \right) \Phi(t,p) \\
           &\times P(r_0, t, \bar{\theta}_r, \bar{m}_r, \bar{\sigma}_r) \, F_\tau^{\nu|p}(dt) \bigg) F_P(dp),
        }
        where
        \formula{
        \Phi(t, p) &= \exp \bigg(-\Lambda_t (1- (\mathcal{L}f_X)((1-\nu)(\alpha p + \beta(1-p))) \\ &+ (1-\nu) \Lambda_t \left(1- \E^{\P_C} \left[ (\mathcal{L}f_X)(\alpha P + \beta(1-P))\right] \right) \bigg).
        }
        and
        \formula{
        F_\tau^{\nu|p}(t) = 1- \sum_{n=0}^\infty \frac{(\Lambda_t^\nu)^n}{n!} \exp(-\Lambda_t^\nu) (F_X^\nu)^{n*}(D_p),
        }
        where $D_p = \min\{D_1/p, D_2/(1-p)\}$ and $\Lambda_t^\nu$, $F_X^\nu$ are described by Proposition~\ref{pla-lt-distribution}.
    \end{enumerate}
    Furthermore, in all three cases we have
    \formula{
    & \bar{\theta}_r = \sqrt{\nu} (\theta_r - \sigma_r \sigma_S 
    \rho (1-\nu)), & &\bar{m}_r = \nu m_r \theta_r/\bar{\theta}_r, & \bar{\sigma}_r = \sqrt{\nu} \sigma_r.
    }
\end{theorem}


\section{Risk-neutral pricing: proofs}
\label{sec:proof}

In order to price multi-region CoCoCat bonds, we redevelop methods introduced in \cite{bgp} for the single-region model. We provide three proofs, each for different assumptions on the loss processes: ILP, ILA, PLA, described in Sections \ref{ilp-definition}, \ref{ila-definition}, \ref{pla-definition}, respectively.

We first notice that the model's equations \eqref{rt}, \eqref{stf}, \eqref{w1w2} and (20), (24), (25) in \cite{bgp} are exactly the same. In other words, in the 2D case we do not change the part of the model corresponding to the financial world. Thus, we will omit these parts in the proof which are identical to the one-dimensional one. Our main goal is to evaluate the expected value $\E I_2$ described in Lemma \ref{is-lemma}.

We divide the present section into three subsections. Although all three proofs may seem similar, their differences lie primarily in the complexity of the loss models. The first model is the simplest, while the last one is the most complex, requiring the use of more advanced methods for valuation. However, all three proofs share common procedural steps, which we describe below.

\textit{Step 1. Finding a risk-neutral measure $\Q$.} This step can be considered as an equivalent of Theorem 3 in \cite{bgp}. As mentioned before, the assumptions regarding the financial word in the 1D and 2D cases are the same, so we only focus on finding the value of $\kappa$ (or $\kappa_1$ and $\kappa_2$ for the independent loss processes) such that $\stc$ is a martingale.

\textit{Step 2. Finding the distribution of the trigger time $\tau$ with respect to $\Q$.} Recall that the trigger time $\tau$ is the first moment at which at least one of the events $\{L_1 \geq D_1\}$ or $\{L_2 \geq D_2\}$ occurs. Clearly, the distribution of $\tau$ depends on the distributions of $L_t^1$ and $L_t^2$, but sometimes it depends solely on the distribution of a certain linear combination of the processes $L_t^1$ and $L_t^2$. Identifying such processes, which we call $\tau$-dependencies, is crucial in the next steps.

\textit{Step 3. Defining a new measure $\P^\nu$ using the Radon--Nikodym derivative.} The most challenging aspect of evaluating $\E^\Q I_2$ is that the process $I_2$ is a certain function of $L_\tau^1$ and $L_\tau^2$, multiplied by the indicator of an event $\{\tau \leq T\}$, and at the same time $\tau$ depends on loss processes. In order to overcome this problem and simply reduce the terms $L_\tau^1$ and $L_\tau^2$ from the expectation, we introduce a new measure.

\textit{Step 4. Identifying the distribution of $\tau$-dependencies with respect to $\P^\nu$.} Since the $\tau$-dependencies found in Step 2 are always compound Poisson processes, the aim of this step is to prove that after change of measure, the $\tau$-dependencies preserve the type of distribution; that is, they are still compound Poisson processes but with different parameters. So we can easily identify the distribution of $\tau$ with respect to $\P^\nu$ simply by substituting the new parameters.

\textit{Step 5. Deriving pricing formula using the change of measure techniques.} We apply the results~of Steps 3 and 4. The new measure $\P^\nu$ introduced in Step 3 allows us to eliminate the problematic terms in $\E^\Q I_2$ and calculate this expectation relatively easily, since the distribution of $\tau$ with respect to $\P^\nu$ is known from Step 4.

While for the first two models (ILP and ILA) this scheme can be applied directly, for the third model (PLA)  major modifications are needed. This is due to the presence of an additional source of randomness, which is the proportion coefficient $P$. We provide a detailed description of these changes in the subsection dedicated to this model. 

\subsection{Independent loss process}\label{ilp-proof}

Recall that under ILP assumption all variables $N_t^1$, $N_t^2$, $X_k^1$, $X_k^2$ are independent, and thus the catastrophic share price process $\stc$ is a product of two independent single-region catastrophic share price processes: $\smash{S_t^{\mathcal{C},1}}$, $\smash{S_t^{\mathcal{C},2}}$ (see \eqref{ilp-stc-product}).

\begin{lemma}\label{ilp-kappa-lemma}
    For ILP and
    \formula[ilp-kappa]{
    & \kappa_1 = \frac{1}{\alpha}(1-(\mathcal{L}f_X)(\alpha)), & \kappa_2 = \frac{1}{\beta}(1-(\mathcal{L}f_X)(\beta))
    }
    there exists a risk-neutral measure $\Q = \Q_F \otimes \P_C$ and the catastrophe-risk and financial market risk variables under this measure are captured by equations \eqref{rt}-\eqref{lts}.
\end{lemma}

\begin{proof}
    Since the financial-world measure is the same as in the 1D case, the second part of the lemma follows for the first part of the proof of Theorem 3 in \cite{bgp}. It only remains to find the values of $\kappa_1, \kappa_2>0$ for which the process $\stc$ is a martingale, that is, the equation $\E^{\P_C}[\stc\,|\,\mathcal{C}_s]=S_s^\mathcal{C}$ is satisfied for all $s<t$. By independence of $\smash{S_t^{\mathcal{C},1}}$ and $\smash{S_t^{\mathcal{C},2}}$, we have
    \formula{
    \E^{\P_C}[\stc\, |\,\mathcal{C}_s] =  \E^{\P_C}[S_t^{\mathcal{C},1} S_t^{\mathcal{C},2}\,|\,\mathcal{C}_s] = \E^{\P_C} [S_t^{\mathcal{C},1} \, | \,\mathcal{C}_s ] \, \E^{\P_C} [S_t^{\mathcal{C},2} \, |\, \mathcal{C}_s ].
    }
    Since $\smash{S_t^{\mathcal{C},1}}$ and $\smash{S_t^{\mathcal{C},2}}$ are catastrophic share prices for the 1D model, by Theorem 3 in \cite{bgp}, for $\kappa_1, \kappa_2$ given by \eqref{ilp-kappa} the processes $\smash{S_t^{\mathcal{C},1}}$ and $\smash{S_t^{\mathcal{C},2}}$ are martingales. Hence,
    \formula{
    \E^{\P_C}[\stc\, |\,\mathcal{C}_s] = S_s^{\mathcal{C},1} S_s^{\mathcal{C},2} = S_s^{\mathcal{C}},
    }
    and the martingale condition for $\smash{\stc}$ is satisfied.
\end{proof}

\begin{proposition}\label{ilp-tau}
    For ILP, under measure $\Q$, trigger time $\tau$ has a distribution function $F_\tau$ given by
    \formula{
    F_\tau(t) = 1 - \exp(-\Lambda_t^1-\Lambda_t^2) \left( \sum_{n=0}^\infty \frac{(\Lambda_t^1)^n}{n!} (F_X^1)^{n*} (D_1) \right) \left( \sum_{n=0}^\infty \frac{(\Lambda_t^2)^n}{n!} (F_X^2)^{n*} (D_2) \right)
    }
    for $t\geq0$.
\end{proposition}

\begin{proof}
    By the independence of $L_t^1$ and $L_t^2$, for $t>0$ we have
    \formula{
    \Q(\tau > t) &= \P_C(\tau > t) \\
    &= \P_C(\tau_1 > t, \, \tau_2 > t ) \\
    &= \P_C(L_t^1 < D_1, \, L_t^2 < D_2) \\
    &= \P_C(L_t^1 < D_1) \, \P_C(L_t^2 < D_2).
    }
    Recall that $L_t^1$ is a compound Poisson process with cumulative intensity $\Lambda_t^1$ and loss distribution function $F_X^1$. Thus,
    \formula{
    \P_C(L_t^1 < D_1) &= \sum_{n=0}^\infty \frac{(\Lambda_t^1)^n}{n!} \exp(-\Lambda_t^1) (F_X^1)^{n*} (D_1).
    }
    This yields the desired formula.
\end{proof}

Note that even tough the case of ILP is the simplest one considering in this paper, the density function of $\tau$, which we further denote as $f_\tau(t) = F'_\tau(t)$, does not follow any neat formula, as opposed to the other two assumptions. Since the calculations are easy, but rather tedious, we omit it here. 

Following the third step, we now define a new measure $\P^\nu$ by Radon-Nikodym derivative. We put
\formula*[ilp-eta]{
\frac{d\P^\nu}{d\P_C} \Bigg\vert_{\hat{\mathcal{C}}_t} &= \exp\big( -\alpha(1-\nu)L_t^1 - \beta(1-\nu)L_t^2 \\ &+ \Lambda_t^1 (1-(\mathscr{L}f_X^1)(\alpha(1-\nu)) + \Lambda_t^2 (1-(\mathscr{L}f_X^2)(\beta(1-\nu)) \big) = \eta(t),
}
where $\eta(t)$ is an exponential martingale. Comparing \eqref{ilp-eta} to (47)-(49) from \cite{bgp}, one can observe that for ILP the transformation kernel $\eta(t)$ is a product of two kernels corresponding to independent single-region catastrophic share price processes. This fact should not be surprising regarding the independence of $L_t^1$ and $L_t^2$.

By Proposition \eqref{ila-tau}, we deduce that the processes $L_t^1$ and $L_t^2$ are $\tau$-dependencies for ILP. Hence, we are interested in finding their distributions under $\P^\nu$.

\begin{proposition}\label{ilp-lt-distribution}
    The processes $L_t^1, L_t^2$ under the measure $\P^\nu$ are compound Poisson process with the same frequency with cumulative intensities
    \formula{
    &\Lambda^{\nu,1}_t = \Lambda_t^1 (\mathcal{L}f_X^1)(\alpha(1-\nu)), &  \Lambda^{\nu,2}_t = \Lambda_t^2 (\mathcal{L}f_X^2)(\beta(1-\nu))&
    }
    and loss distribution functions
    \formula{
    & F_X^{\nu,1}(x) = \frac{\exp(-\alpha(1-\nu)x)}{(\mathcal{L}f_X^1)(\alpha(1-\nu))} F_X^1(x), & F_X^{\nu,2}(x) = \frac{\exp(-\beta(1-\nu)x)}{(\mathcal{L}f_X^2)(\beta(1-\nu))} F_X^2(x). &
    }
\end{proposition}
\begin{proof}
    We prove the assertion by comparing moment generating functions. Recall that for compound Poisson process $L_t^1$ and $z \in \R$ we have
    \formula[ilp-mgf-pc]{
    \E^{\P_C} \exp(-z L_t^1) = \exp ( ((\mathcal{L}f_X^1)(z) -1) \Lambda_t^1). 
    }
    By the formula \eqref{ilp-eta} and independence of $L_t^1$ and $L_t^2$,
    \formula{
    \E^{\P^\nu} \exp(-z L_t^1) &= \E^{\P_C}  \left[ \exp(-z L_t^1 ) \, \eta(t) \right] \\
    &= \E^{\P_C} \exp( -(z + \alpha(1-\nu))L_t^1 - \beta(1-\nu) L_t^2)  \\
    & \times \exp(\Lambda_t^1 (1-(\mathscr{L}f_X^1)(\alpha(1-\nu)) + \Lambda_t^2 (1-(\mathscr{L}f_X^2)(\beta(1-\nu)) ) \\
    &= \exp( ((\mathscr{L}f_X^1)(z + \alpha(1-\nu))-1) \Lambda_t^1 ) \\
    & \times \exp( ((\mathscr{L}f_X^2)(\beta(1-\nu))-1) \Lambda_t^2 ) \\
    & \times \exp(\Lambda_t^1 (1-(\mathscr{L}f_X^1)(\alpha(1-\nu)) + \Lambda_t^2 (1-(\mathscr{L}f_X^2)(\beta(1-\nu)) ) \\
    &= \exp(\Lambda_t^1((\mathscr{L}f_X^1)(z + \alpha(1-\nu)) - (\mathscr{L}f_X^1)(\alpha(1-\nu)))) \\
    &= \exp \left( \Lambda_t^1 (\mathscr{L}f_X^1)(\alpha(1-\nu)) \left( \frac{(\mathscr{L}f_X^1)(z+\alpha(1-\nu))}{(\mathscr{L}f_X^1)(\alpha(1-\nu))}-1\right)\right).
    }
    Comparing the above formula to \eqref{ilp-mgf-pc}, we obtain
    \formula{
    \Lambda_t^{\nu,1} = \Lambda_t^1 (\mathscr{L}f_X^1)(\alpha(1-\nu)).
    }
    Moreover,
    \formula{
    \frac{(\mathscr{L}f_X^1)(z+\alpha(1-\nu))}{(\mathscr{L}f_X^1)(\alpha(1-\nu))} = \int_0^\infty e^{-xz} \frac{e^{-\alpha(1-\nu)x} F(dx)}{(\mathscr{L}f_X^1)(\alpha(1-\nu))} = \int_0^\infty e^{-xz} F_X^{\nu,1}(dx) = (\mathscr{L}F_X^{\nu,1})(z).
    }
    and hence we read that $F^{\nu,1}_X$ has the desired form. The remaining part of the proof for $L_t^2$ is identical.
\end{proof}

Combining Propositions \ref{ilp-tau} and \ref{ilp-lt-distribution}, we find that distribution function $F_\tau^\nu$ of $\tau$ with respect to the measure $\P^\nu$ is given by the same formula but with $\Lambda_t^{\nu, i}$ and $F_X^{\nu, i}$ instead of $\Lambda_t^i$ and $F_X^i$ for $i=1,2$. With these tools, we are now ready to evaluate the expected value of $I_2$. Omitting the constants, by \eqref{stf-formula}, \eqref{stc-formula}, \eqref{b0t}, we have
\formula*[ilp-expectation1]{
\E^\Q \big[ S_\tau^{1-\nu} B(0,\tau) &\ind_{\tau \leq T} \big] \\
&= \E^\Q \bigg[ \exp\bigg( -\alpha (1-\nu) L_\tau^1 - \beta (1-\nu) L_\tau^2 - \nu \int_0^\tau r_u \, du \\ &+ (1-\nu) \sigma_S W_\tau^2+ (1-\nu) \left( \alpha \kappa_1 \Lambda_\tau^1 + \beta \kappa_2 \Lambda_\tau^2 - \tfrac{1}{2} \tau \sigma_S^2 \right) \bigg) \ind_{\tau \leq T} \bigg].
}

We now define a new product measure $\overline{\Q} = \P^\nu \otimes \Q_F$ such that for any $A \in \hat{\mathcal{C}}_t$ and $B \in \hat{\mathcal{F}}_t$ we have
\formula{
\E^{\overline{\Q}} \ind_{A \times B} = \E^{\P^\nu} \ind_A \, \E^{\Q_F} \ind_B = \E^{\P^\nu}[\eta(t) \ind_A] \, \E_F \ind_B =  \E^{\Q}[\eta(t) \ind_{A \times B}].
}
It follows that
\formula{
\frac{d\overline{\Q}}{d\Q} \bigg\vert_{\mathcal{G}_t} = \eta(t),
}
where $\eta(t)$ is defined by \eqref{ilp-eta}. In the above equation one can change deterministic $t$ to a stopping time $\tau$, provided that $\tau < \infty$ (see Proposition 1.7.1.4 in \cite{jyc}). Thus, \eqref{ilp-expectation1} is equal to
\formula[ilp-expectation2]{
\E^{\overline{\Q}} \bigg[ \varphi(\tau) \exp\bigg( - \nu \int_0^\tau r_u \, du + (1-\nu) \sigma_S W_\tau^2  \bigg) \ind_{\tau \leq T}
\bigg],
}
where
\formula{
\varphi(\tau) = \exp \bigg( -\Lambda_\tau^1 (1-(\mathscr{L}f_X^1)(\alpha(1-\nu)) &- \Lambda_\tau^2 (1-(\mathscr{L}f_X^2)(\beta(1-\nu)) \\
&+ (1-\nu) \left( \alpha \kappa_1 \Lambda_\tau^1 + \beta \kappa_2 \Lambda_\tau^2 - \tfrac{1}{2} \tau \sigma_S^2 \right) \bigg).
}
Denoting the density of $\tau$ under $\overline{\Q}$ (or under $\P^\nu$) by $f_\tau^\nu$ and conditioning \eqref{ilp-expectation2} with $\tau$, we obtain
\formula*[ilp-expectation3]{
\E^{\overline{\Q}} \bigg[ \varphi(\tau) \, \E^{\overline{\Q}} &\bigg[ \exp\bigg( - \nu \int_0^\tau r_u \, du + (1-\nu) \sigma_S W_\tau^2 \bigg) \ind_{\tau \leq T}
\, | \, \tau \bigg] \bigg] \\
&= \int_0^T \varphi(t) \, \E^{\Q_F} \left[ \exp\left( - \nu \int_0^t r_u \, du + (1-\nu) \sigma_S W_t^2 \right) \right] f_\tau^\nu(t) \, dt.
}
By the arguments presented in \cite{bgp} (see Lemma 2 and equations (55)-(60) therein), the above expectation with respect to the financial measure $\Q_F$ can be simplified to
\formula{
\exp\left( \tfrac{1}{2}(1-\nu)^2 \sigma_S^2 t \right) P(r_0, t, \bar{\theta}_r, \bar{m}_r, \bar{\sigma}_r),
}
where
\formula[params]{
    & \bar{\theta}_r = \sqrt{\nu} (\theta_r - \sigma_r \sigma_S 
    \rho (1-\nu)), & &\bar{m}_r = \nu m_r \theta_r/\bar{\theta}_r, & \bar{\sigma}_r = \sqrt{\nu} \sigma_r.
}
Therefore \eqref{ilp-expectation3} can be expressed as
\formula{
\int_0^t \varphi(t) \exp\left( \tfrac{1}{2}(1-\nu)^2 \sigma_S^2 t \right) P(r_0, t, \bar{\theta}_r, \bar{m}_r, \bar{\sigma}_r) f_\tau^\nu(t) \, dt,
}
which proves the theorem for ILP case. \qed

\subsection{Independent loss amounts}\label{ila-proof}

The second type of aggregated loss process that we consider is ILA, in which the loss amounts are independent with common counting process $N_t$. Merged aggregate loss process $L_t$ defined by \eqref{ila-lt-merged} obeys compound Poisson distribution with loss amounts
\formula{
X_k = \frac{\alpha}{\alpha + \beta} X_k^1 + \frac{\beta}{\alpha+\beta} X_k^2.
}
It is easy to see that if $F_X^1$ and $F_X^2$ are distribution functions of $X_1$ and $X_2$, respectively, then the distribution function $F_X$ of $X_k$ is given by the convolution 
\formula[ila-F]{
F_X(x) = (\tilde{F}_X^1 * \tilde{F}_X^2)(x),
}
where
\formula{
& \tilde{F}_X^1(x) = F_X^1 \left( \frac{\alpha x}{\alpha+\beta} \right), & \tilde{F}_X^2(x) = F_X^2 \left( \frac{\beta x}{\alpha+\beta} \right). &
}
In this case, catastrophic share price process \eqref{stc} reads
\formula{
\stc &= \exp \left( -(\alpha+\beta)L_t + (\alpha+\beta) \kappa \int_0^t \lambda_u \, du \right).
}

The proof for ILA also follows the steps described at the beginning of this section. Below we present a few results crucial in proving the pricing formula.

\begin{lemma}
    For ILA and
    \formula[ila-kappa]{
    \kappa = \frac{1}{\alpha+\beta} \left(1- (\mathcal{L}f_X^1)(\alpha) (\mathcal{L}f_X^2)(\beta) \right)
    }
    there exists a risk-neutral measure $\Q = \Q_F \otimes \P_C$ and the catastrophe-risk and financial market risk variables under this measure are captured by equations \eqref{rt}-\eqref{lts}.
\end{lemma}

\begin{proof}
    Inserting $\alpha+\beta$ instead of $\alpha$ and $F_X$ given by \eqref{ila-F} into Theorem 3 in \cite{bgp} leads to our claim with
    \formula{
    \kappa = \frac{1}{\alpha+\beta}(1-(\mathcal{L}f_X)(\alpha+\beta)).
    }
    Recall that $X_k = \frac{\alpha}{\alpha + \beta} X_k^1 + \frac{\beta}{\alpha+\beta} X_k^2$ and by that,
    \formula{
    (\mathcal{L} f_X)(\alpha+\beta) &= \E^\Q \exp\left( -(\alpha+\beta) X_k \right) \\ &= \E^\Q \exp\left( -\alpha X_k^1 - \beta X_k^2 \right) = (\mathcal{L}f_X^1)(\alpha) (\mathcal{L} f_X^2)(\beta).
    }
    This concludes the proof.
\end{proof}

\begin{proposition}\label{ila-tau}
    For independent loss amounts, under measure $\Q$, trigger time $\tau$ has a distribution function $F_\tau$ and density $f_\tau$ given by
    \formula{
    F_\tau(t) &= 1- \sum_{n=0}^\infty \frac{\Lambda_t^n}{n!} \exp(-\Lambda_t) (F_X^1)^{n*} (D_1)(F_X^2)^{n*} (D_2), \\
    f_\tau(t) &= \lambda_t \exp(-\Lambda_t) \left( 1-\sum_{n=1}^\infty \left( \frac{\Lambda_t^{n-1}}{(n-1)!} - \frac{\Lambda_t^n}{n!} \right) (F_X^1)^{n*} (D_1)(F_X^2)^{n*} (D_2) \right)
    }
    for $t\geq0$.
\end{proposition}

\begin{proof}
    Let $t >0$. Observe that $\tau$ does not depend on financial world and only on catastrophic world, hence
    \formula{
    \Q (\tau > t) &= \P_C  (\tau > t) \\
    &= \P_C(\tau_1 > t, \, \tau_2 > t) \\
    &= \P_C\left( L_t^1 < D_1, \, L_t^2 < D_2 \right) \\
    &= \E^{\P_C} \left[ \P_C\left( L_t^1 < D_1, \, L_t^2 < D_2 \, | \, N_t \right) \right] \\
    &=  \E^{\P_C} \left[ \P_C\left( L_t^1 < D_1 \, | N_t \,\right) \P_C \left( L_t^2 < D_2 \, | \, N_t \right) \right] \\
    &= \E^{\P_C} \left[ (F_X^1)^{N_t*}(D_1) (F_X^2)^{N_t*}(D_2) \right].
    }
    Above we used the fact that $L_t^1, L_t^2$ are independent conditionally on $N_t$. Since $N_t$ has a Poisson distribution with cumulative intensity $\Lambda_t$, we have
    \formula{
    \Q(\tau > t) = \sum_{n=0}^\infty \frac{\Lambda_t^n}{n!} \exp(-\Lambda_t) (F_X^1)^{n*}(D_1) (F_X^2)^{n*}(D_2),
    }
    Differentiating the above formula with respect to $t$ yields 
    \formula{
    \lambda_t \exp(-\Lambda_t) \left( \sum_{n=1}^\infty \left( \frac{\Lambda_t^{n-1}}{(n-1)!} - \frac{\Lambda_t^n}{n!} \right) (F_X^1)^{n*}(D_1) (F_X^2)^{n*}(D_2) -1 \right) 
    }
    and the assertion follows.
\end{proof}

The next part of the proof is also analogical; we define a new measure $\P^\nu$ by its Radon--Nikodym derivative:
\formula*[ila-eta]{
\frac{d \P^\nu}{d \P_C} \Bigg\vert_{\hat{\mathcal{C}}_t} &= \exp\bigg(-(1-\nu) (\alpha+\beta) L_t \\ &\quad+ \Lambda_t(1- (\mathcal{L}f_X^1)(\alpha (1-\nu)) (\mathcal{L}f_X^2)(\beta (1-\nu))) \bigg) = \eta(t),
}
where $\eta(t)$ is again an exponential martingale.

Similarly to the ILP case, for ILA the $\tau$-dependencies indicated by Proposition \ref{ila-eta} are again aggregated loss processes $L_t^1$ and $L_t^2$. The following result fulfils the fourth step of the proof.

\begin{proposition}\label{ila-lt-distribution}
    The processes $L_t^1, L_t^2$ under the measure $\P^\nu$ are compound Poisson process with the same frequency with cumulative intensity
    \formula{
    \Lambda^{\nu}_t = \Lambda_t (\mathcal{L}f_X^1)(\alpha(1-\nu)) (\mathcal{L}f_X^2)(\beta(1-\nu))
    }
    and loss distribution functions
    \formula{
    & F_X^{\nu,1}(x) = \frac{\exp(-\alpha(1-\nu)x)}{(\mathcal{L}f_X^1)(\alpha(1-\nu))} F_X^1(x), & F_X^{\nu,2}(x) = \frac{\exp(-\beta(1-\nu)x)}{(\mathcal{L}f_X^2)(\beta(1-\nu))} F_X^2(x). &
    }
\end{proposition}

\begin{proof}
    We check the distributions directly by studying the moment generating functions. Recall that for the compound Poisson process $L_t^1$ and $z \in \R$ we have
    \formula[ila-mgf-pc]{
    \E^{\P_C} \exp(-z L_t^1) = \exp ( ((\mathcal{L}f_X^1)(z) -1) \Lambda_t). 
    }
    By the change of measure \eqref{ila-eta},
    \formula*[ila-mgf-pnu]{
    \E^{\P^\nu} \exp(-z L_t^1) &= \E^{\P_C}  \left[ \exp(-z L_t^1 ) \eta(t) \right] \\
    &= \exp((1- (\mathcal{L}f_X^1)((1-\nu)(\alpha)) (\mathcal{L}f_X^2)((1-\nu)(\beta))) \Lambda_t) \\
    & \times \E^{\P_C} \exp\left( -(z+\alpha(1-\nu)) L_t^1 - \beta(1-\nu) L_t^2 \right).
    }
    Recall that $L_t^1$ and $L_t^2$ are independent conditionally on $N_t$, hence the latter expectation is equal to
    \formula{
    \E^{\P_C} &\left[ \E^{\P_C} \left[\exp\left( -(z+\alpha(1-\nu)) L_t^1 - \beta(1-\nu) L_t^2 \right) \, | \, N_t \right] \right]
    \\
    &=\E^{\P_C} \bigg[ \E^{\P_C} \left[\exp ( -(z+\alpha(1-\nu)) L_t^1 ) \, | \, N_t \right] \, \E^{\P_C} \left[ \exp(- \beta(1-\nu) L_t^2 ) \, | \, N_t \right] \bigg] \\
    &= \E^{\P_C} \bigg[ \prod_{k=1}^{N_t} \E^{\P_C} \left[\exp ( -(z+\alpha(1-\nu)) X_1^1 )\right] \, \E^{\P_C} \left[\exp(- \beta(1-\nu) X_1^2 )\right] \bigg] \\
    &= \E^{\P_C} \bigg[ \prod_{k=1}^{N_t} (\mathcal{L}f_X^1)(z+\alpha(1-\nu)) (\mathcal{L}f_X^2)(\beta(1-\nu)) \bigg] \\
    &= \E^{\P_C} \bigg[ \big( (\mathcal{L}f_X^1)(z+\alpha(1-\nu)) (\mathcal{L}f_X^2)(\beta(1-\nu)) \big)^{N_t} \bigg].
    }
    We observe that the latter expected value above is a probability generating function of $N_t$ evaluated in $(\mathcal{L}f_X^1)(z+\alpha(1-\nu)) (\mathcal{L}f_X^2)(\beta(1-\nu))$. Hence, this is equal to
    \formula{
    \exp\big( ((\mathcal{L}f_X^1)(z+\alpha(1-\nu)) (\mathcal{L}f_X^2)(\beta(1-\nu)) -1) \Lambda_t \big).
    }
    Inserting this into \eqref{ila-mgf-pnu} gives
    \formula{
    \E^{\P^\nu} &\exp(-z L_t^1) \\
    &= \exp \big( (\mathcal{L}f_X^1)(z+\alpha(1-\nu)) (\mathcal{L}f_X^2)(\beta(1-\nu)) \Lambda_t \\
    &- (\mathcal{L}f_X^1)((1-\nu)(\alpha)) (\mathcal{L}f_X^2)((1-\nu)\beta)) \Lambda_t \big) \\
    &= \exp \left( \Lambda_t (\mathcal{L}f_X^1)(\alpha(1-\nu)) (\mathcal{L}f_X^2)(\beta(1-\nu)) \left( \frac{(\mathcal{L}f_X^1)(z+\alpha(1-\nu))}{(\mathcal{L}f_X^1)(\alpha(1-\nu))} -1 \right) \right).
    }
    We obtain our claim for $L_t^1$ by comparing the parameters with \eqref{ila-mgf-pc}. By similar arguments, one can prove an analogous result for $L_t^2$. 
\end{proof}

The rest of the calculations can be directly repeated from the proof for ILP. We change the product measure $\Q$ for the new one $\smash{\overline{\Q} = \P^\nu \otimes \Q_F}$ and after analogous transformations we obtain
\formula{
\E^\Q \big[ S_\tau^{1-\nu} B(0,\tau) &\ind_{\tau \leq T} \big] = \int_0^t \varphi(t) \exp\left( \tfrac{1}{2}(1-\nu)^2 \sigma_S^2 t \right) P(r_0, t, \bar{\theta}_r, \bar{m}_r, \bar{\sigma}_r) f_\tau^\nu(t) \, dt,
}
where the density $f_\tau^\nu$ of $\tau$ under $\P^\nu$ follows from replacing $\Lambda_t$ and $F_X$ in density $f_\tau$ from Proposition~\ref{ila-tau} with $\Lambda_t^\nu$ and $F_X^\nu$ from Proposition~\ref{ila-lt-distribution}, parameters $\bar{\theta}$, $\bar{m}_r$, $\bar{\sigma}_r$ are defined as in \eqref{params} and
\formula{
\varphi(t) = \exp \bigg( -\Lambda_t(1- &(\mathcal{L}f_X^1)(\alpha (1-\nu)) (\mathcal{L}f_X^2)(\beta (1-\nu))) \\
&+ (1-\nu) \left( (\alpha + \beta) \kappa \Lambda_t - \tfrac{1}{2} t \sigma_S^2 \right) \bigg).
}
This finishes the proof for ILA. \qed

\subsection{Proportional loss amounts} Now we will focus on the model with proportional loss amounts. Recall that in this case the catastrophic share price is given by the process
\formula[pla-stc]{
\stc = \exp\left( \alpha \sum_{k=1}^{N_t} P X_k + \beta \sum_{k=1}^{N_t} (1-P) X_k + \kappa (\alpha+ \beta) \int_0^t \lambda_u \, du \right),
}
where $N_t, P, X_k$ are all independent. Following the notation introduced in Section \ref{pla-definition}, let $Q = \alpha P + \beta(1-P)$ and $L_t = \sum_{k=1}^{N_t} X_k$ so that
\formula{
\alpha \sum_{k=1}^{N_t} P X_k + \beta \sum_{k=1}^{N_t} (1-P) X_k = Q L_t,
}

As before, we will start with finding a risk-neutral measure for this model. We state our result as a separate lemma. 

\begin{lemma}
    For RPL and
    \formula[pla-kappa]{
    \kappa = \frac{1}{\alpha+\beta} \left(1- \E^{\P_C} \left[ (\mathcal{L}f_X)(Q)\right] \right)
    }
    there exists a risk-neutral measure $\Q = \Q_F \otimes \P_C$ and the catastrophe-risk and financial market risk variables under this measure are captured by equations \eqref{rt}-\eqref{lts}.
\end{lemma}

\begin{proof}
Since the financial-world measure is the same as in 1-D case, the second part of the lemma follows for the first part of the proof of Theorem 3 in \cite{bgp}. It remains only to find the value of $\kappa>0$ for which the process $\stc$ is a martingale, that is the equation $\E^{\P_C}[\stc|\mathcal{C}_s]=S_s^\mathcal{C}$ is satisfied for all $s<t$.

This condition can be rewritten as 
\formula[martingale-condition]{
    \E^{\P_C} \left[ \exp\left( - Q (L_t-L_s) + \kappa (\alpha+\beta) \int_s^t \lambda_u \, du  \right) | \, \mathcal{C}_s \right] = 1.
}
Let us inspect expectation $\E^{\P_C} \exp(-Q(L_t-L_s))$. Conditioning twice, we obtain
\formula{
\E^{\P_C} \exp(-Q(L_t-L_s))
&= \E^{\P_C} \left[ \E^{\P_C}[\E^{\P_C}[\exp(-Q(L_t-L_s)) \,| \, Q ] | \, N_t - N_s] ] \right] \\
&= \E^{\P_C} \left[ \E^{\P_C}\left[\E^{\P_C} \left[ \exp\left(-Q \sum_{k=1}^{N_t-N_s} X_k\right) |\, Q \right] | \, N_t - N_s \right] \right] \\
&= \E^{\P_C} \left[ \prod_{k=1}^{N_t-N_s} \E^{\P_C}\left[ \E^{\P_C} \left[ \exp\left(-Q X_k \right) | \, Q \right] \right] \right] \\
&= \E^{\P_C} \left[ \prod_{k=1}^{N_t-N_s} \E^{\P_C} [(\mathcal{L}f_X)(Q)]  \right] \\
&= \E^{\P_C}\left[ (\E^{\P_C} [(\mathcal{L}f_X)(Q)])^{N_t-N_s}  \right] \\
&= G_{N_t-N_s} (\E^{\P_C} [(\mathcal{L}f_X)(Q)]) \\
&= \exp \left( (\E^{\P_C} [(\mathcal{L}f_X)(Q)]-1) \int_s^t \lambda_u \, du \right).
}
where $G_{N_t-N_s}$ is the probability generating function of the Poisson random variable with mean $\int_s^t \lambda_u \, du$. Hence, \eqref{martingale-condition} is satisfied if and only if
\formula{
\kappa = \frac{1}{\alpha+\beta} \left(1- \E^{\P_C}\left[ (\mathcal{L}f_X)(Q)\right] \right)
}
and the proof is complete.
\end{proof}

Note that when we assume fixed proportional loss amounts with deterministic proportion coefficient $p$ instead of random $P$, the above result remains true with $q = \alpha p + \beta(1-p)$ instead of $Q$. Since $q$ is  deterministic, equation \eqref{pla-kappa} simplifies to
\formula{
\kappa = \frac{1}{\alpha + \beta} (1-(\mathcal{L}f_X)(q)).
}

Before we proceed to the next step of the proof, let us explain the idea for RPL. Recall that our aim is to calculate $\E^\Q \big[ S_\tau^{1-\nu} B(0,\tau) \ind_{\tau \leq T} \big]$, where, unlike for ILP and ILA, the process $S_t$ depends also $P$. That's why we condition this expected value with $P$ to get
\formula[pla-expectation1]{
\E^\Q \left[ \E^\Q \left[ S_\tau^{1-\nu} B(0,\tau) \ind_{\tau \leq T} \, |\, P \right] \right].
}
In this case, the trick with change of measure will be applied in the inner (conditional) expectation. We define a conditional catastrophic measure $\P_{C|P}(A)$ such that for any $\smash{A \in \hat{\mathcal{C}}_\infty}$ we have
\formula{
\P_{C|P}(A) = \P_C(A|P).
}
Similarly, let $\Q_{|P} = \P_{C|P} \otimes \Q_F$ be a corresponding product measure. We can now rewrite \eqref{pla-expectation1} as
\formula[pla-expectation2]{
\E^\Q \left[ \E^{\Q_{|P}} \left[ S_\tau^{1-\nu} B(0,\tau) \ind_{\tau \leq T}\right] \right]
}
and continue to work with the inner expectation according to a similar scheme as for ILP and ILA but with $\Q_{|P}$ instead of $\Q$.

Proceeding to the second step of the proof, we identify the distribution of the time of the trigger time $\tau$. The results are given in the following proposition.

\begin{proposition}\label{pla-tau}
    For random proportional loss process, under measure $\Q_{|P}$, trigger time $\tau$ has a distribution function $F_\tau$ and density $f_\tau$ given by
    \formula{
    F_\tau(t) &= 1- \sum_{n=0}^\infty \frac{\Lambda_t^n}{n!} \exp(-\Lambda_t) F_X^{n*}(D_P)  \\
    f_\tau(t) &= \lambda_t \exp(-\Lambda_t) \left( 1-\sum_{n=1}^\infty \left( \frac{\Lambda_t^{n-1}}{(n-1)!} - \frac{\Lambda_t^n}{n!} \right) F_X^{n*}(D_P) \right)
    }
    for $t \geq 0$, where $D_P = \min\left\{D_1/P, D_2/(1-P)\right\}$.
\end{proposition}

\begin{proof}
    Let $t >0$. As before, we have
    \formula{
    \Q_{|P} (\tau > t) &= \P_C  (\tau > t\, |\,P) \\
    &= \P_C(\tau_1 > t, \, \tau_2 > t \, | \, P) \\
    &= \P_C\left( L_t^1 < D_1, \, L_t^2 < D_2 \, | \, P\right) \\
    &= \P_C\left( P L_t < D_1, \, (1-P) L_t < D_2 \, | \, P\right) \\
    &= \P_C\left( L_t < D_P \, |\, P\right).
    }
    where $L_t = \sum_{k=1}^{N_t} X_k$ and $D_P = \min\left\{D_1/P, D_2/(1-P)\right\}$. Recall that $N_t$ follows a Poisson distribution with cumulative intensity $\Lambda_t$ and $(X_k)$ are i.i.d. random variables, independent from $N_t$, with distribution function $F_X$. Hence,
    \formula{
    \P_C (L_t < D_P \, | \, P) &= \P_C(L_t = 0 \, | \, P) + \P_C(0 < L_t < D_P \, | \, P) \\
    &= \exp(-\Lambda_t) + \sum_{n=1}^\infty \frac{\Lambda_t^n}{n!} \exp(-\Lambda_t) F_X^{n*}(D_P),
    }
    Differentiating the above formula with respect to $t$ yields 
    \formula{
    \lambda_t \exp(-\Lambda_t) \left( \sum_{n=1}^\infty \left( \frac{\Lambda_t^{n-1}}{(n-1)!} - \frac{\Lambda_t^n}{n!} \right) F_X^{n*}(D_P) -1 \right) 
    }
    and the assertion follows.
\end{proof}

Note that this result is also useful for calculating the probabilities that appear in \eqref{ei1} or \eqref{ei3}. More precisely, for $t \geq 0$ we have
\formula{
\Q(\tau > t) &= \E^\Q \left[ \Q (\tau > t \, | \, P) \right] \\
&= \int_0^1 \left( \sum_{n=0}^\infty \frac{\Lambda_t^n}{n!} \exp(-\Lambda_t) F_X^{n*}(D_p) \right) F_P(dp)
}

The third step is to find a new conditional measure which we will denote as $\P^{\nu|P}$. Let
\formula[pla-eta]{
\frac{d \P^{\nu|P}}{d \P_{C|P}} \Bigg\vert_{\hat{\mathcal{C}}_t} = \exp(-(1-\nu) Q L_t + \Lambda_t (1- (\mathcal{L}f_X)((1-\nu)Q))) = \eta(t).
}
It is easy to verify that the process $\eta(t)$ is an exponential $\P^{\nu|P}$-martingale.

By Proposition~\ref{pla-tau}, we also find that that the process $L_t$ is a $\tau$-dependency under the conditional measure. As a fourth step, we prove that $L_t$ preserves its distribution with respect to the new measure, as it was in previous cases.

\begin{proposition}\label{pla-lt-distribution}
    The process $L_t$ under the measure $\P^{\nu|P}$ is a compound Poisson process with parameters
    \formula{
    & \Lambda^{\nu|P}_t = \Lambda_t (\mathcal{L}f_X)((1-\nu)Q),
    & F_X^{\nu|P}(x) = \frac{\exp(-(1-\nu)Qx)}{(\mathcal{L}f_X)((1-\nu)Q)} F_X(x).
    }
\end{proposition}

\begin{proof}
    Our claim easily follows from Proposition~\ref{ilp-lt-distribution} with $\alpha = Q$ and $\beta = 0$.
\end{proof}

The final step is similar to the ILP and ILA case. We define a new measure $\overline{Q}_{|P} = \P^{\nu|P} \otimes \Q_F$ such that
\formula{
\frac{d\overline{\Q}}{d\Q} \bigg\vert_{\mathcal{G}_t} = \eta(t),
}
where $\eta(t)$ is defined by \eqref{pla-eta}. Then we examine the inner expectation value in~\eqref{pla-expectation2}, which yields
\formula{
\E^{\Q_{|P}} &\left[ S_\tau^{1-\nu} B(0,\tau) \ind_{\tau \leq T}\right] \\
&= \E^{\Q_{|P}} \bigg[ \exp\bigg( -(1-\nu)Q L_t - \nu \int_0^\tau r_u \, du +(1-\nu) \sigma_S W_\tau^2 \\ & \quad + (1-\nu) \left( \alpha \kappa_1 \Lambda_\tau^1 + \beta \kappa_2 \Lambda_\tau^2 - \tfrac{1}{2} \tau \sigma_S^2 \right) \bigg) \ind_{\tau \leq T} \bigg] \\
&= \E^{\overline{Q}_{|P}} \bigg[ \varphi (\tau, Q) \exp\bigg( - \nu \int_0^\tau r_u \, du +(1-\nu) \sigma_S W_\tau^2  \bigg) \ind_{\tau \leq T} \bigg],
}
where
\formula{
\varphi(t, q) = \exp \left(-\Lambda_t (1- (\mathcal{L}f_X)((1-\nu)q)) + (1-\nu) \left( (\alpha + \beta) \kappa \Lambda_t - \tfrac{1}{2} t \sigma_S^2 \right) \right).
}
Conditioning with $\tau$ gives us
\formula{
\int_0^T \varphi(t, Q) \, \E^{\Q_F} \left[ \exp\bigg( - \nu \int_0^\tau r_u \, du +(1-\nu) \sigma_S W_\tau^2  \bigg) \right] f_\tau^{\nu|P}(t) \, dt.
}
This expression is very similar with these obtained in the previous cases. Since no changes considering the financial measure were made and financial and catastrophic worlds are independent, we can write
\formula{
\E^{\Q_F} \left[ \exp\bigg( - \nu \int_0^\tau r_u \, du +(1-\nu) \sigma_S W_\tau^2  \bigg) \right] = \exp\left( \tfrac{1}{2}(1-\nu)^2 \sigma_S^2 t \right) P(r_0, t, \bar{\theta}_r, \bar{m}_r, \bar{\sigma}_r),
}
as before.

Regarding \eqref{pla-expectation2} and the above results, we obtain the final pricing formula:
\formula{
\int_0^1 \left( \int_0^T \varphi(t, q) \, \exp\left( \tfrac{1}{2}(1-\nu)^2 \sigma_S^2 t \right) P(r_0, t, \bar{\theta}_r, \bar{m}_r, \bar{\sigma}_r) \, f_\tau^{\nu|p}(t) \, dt \right) F_P(dp),
}

where $f_\tau^{\nu|p}$ follows the same formula to $f_\tau$ in Proposition~\ref{pla-tau} with $P=p$, and $\Lambda_t^\nu$ and $F_X^\nu$ instead of $\Lambda_t$ and $F_X$.

Note that if we consider fixed proportional loss amounts, that is when $P=p \in (0,1)$ with probability one, the above formula simplifies; for example the outer integral vanishes. We leave the details to the reader. \qed


\section{Further possible generalizations}
\label{sec:gen}

Regarding the results and methods of pricing gathered in this paper, one can generalise considered model for losses from three regions and more. It can be easily done for ILP and ILA. However, under the assumption of PLA, additional technical difficulties appear (they are primarily related to the proportion coefficients). In the following, we provide a result for the two former assumptions. 

Consider $R$ different regions, where $R=1,2,3,\ldots$, and $R$ aggregate loss processes defined as
\formula{
& L_t^r = \sum_{k=1}^{N_t^r} X_k^r & \text{for }r=1,2,\ldots, R, &
}
where $X_1^r, X_2^r, X_3^r, \ldots$ is a sequence of i.i.d. positive, continuous random variables with distribution functions $F_X^r$ and densities $f_X^r$, independent from the Poisson process $N_t^r$ with cumulative intensities $\smash{\Lambda_t^r = \int_0^t \lambda_u^r \, du}$.

Under ILP, the processes $L_t^r$ and $L_t^s$ are independent for different $r,s$ in the sense of pairwise independence of any two different random components therein. Under ILA, we assume that the processes $N_t^r$ are equal for every $r$ and we denote them by $N_t$ and their cumulative intensity as $\Lambda_t$.

For positive thresholds $D_1, D_2, \ldots, D_R$ we put
\formula{
& \tau_r = \min\{t \geq 0: L_t^r \geq D_r \}, & \text{for }r=1,2,\ldots, R&
}
and the trigger time for the multi-region CoCoCat bond is defined as
\formula{
\tau = \min\{\tau_r: r = 1,2,\ldots, R\}.
}

\begin{theorem}\label{thm-general}
    Risk-neutral price of multi-region CoCoCat for $R$ regions is given by
    \formula{
    \E^\Q I_1 + \E^\Q I_2 + \E^\Q I_3.
    }
    The first and third component are described by \eqref{ei1} and \eqref{ei3} while the second reads
    \formula{
    \E^\Q I_2 = \zeta Z S_0^{1-\nu} \int_0^\infty \exp\left(-\tfrac{1}{2} \nu (1-\nu)t \sigma_S^2\right) \Phi(t) P(r_0, t, \bar{\theta}_r, \bar{m}_r, \bar{\sigma}_r) F_\tau^\nu(dt),
    }
    where ${\bar{\theta}}_r, \bar{m}_r, \bar{\sigma}_r$ are described in Theorem \ref{thm-main}. Moreover,
    \formula{
    \Phi(t) = \begin{cases}
        \displaystyle{\prod_{r=1}^R \exp\bigg( -\Lambda_t^r\big(1-(\mathcal{L}f_X^r)(\alpha_r(1-\nu))) -  (1-\nu)(1-(\mathcal{L}f_X^r)(\alpha_r))\big) \bigg)} & \text{for ILP}, \\
        \displaystyle{\exp \left( -\Lambda_t (1-\nu) \left(1-\prod_{r=1}^R (\mathcal{L}f_X^r)(\alpha_r) \right) - \prod_{r=1}^R (\mathcal{L}f_X^r)(\alpha_r(1-\nu)) \right)} &\text{for ILA},
    \end{cases}
    }
    and
    \formula{
    F_\tau^\nu(t) = \begin{cases}
        \displaystyle{
        1 - \prod_{r=1}^R \left( \exp(-\Lambda_t^{\nu, r}) \sum_{n=0}^\infty \frac{(\Lambda_t^{\nu,r})^n}{n!} (F_X^{\nu,r})^{n*}(D_r) \right)
        } &\text{for ILP}, \\
        \displaystyle{
        1-\exp(-\Lambda_t^\nu) \sum_{n=0}^\infty \left( \frac{(\Lambda_t^\nu)^n}{n!} \prod_{r=1}^R (F_X^{\nu,r})^{n*}(D_r) \right)
        } &\text{for ILA},
    \end{cases}
    }
   where 
    \formula{
    & \Lambda_t^{\nu,r} = \Lambda_t^r (\mathcal{L}f_X^r)(\alpha_r(1-\nu)), & \Lambda^\nu_t = \Lambda_t \prod_{r=1}^R (\mathcal{L}f_X^r)(\alpha_r(1-\nu)) &
    }
    and
    \formula{
        F_X^{\nu, r} &= \frac{\exp(-\alpha_r(1-\nu)x)}{(\mathcal{L}f_X^r)(\alpha_r(1-\nu))} F_X^r(x),
    }
    for $r=1,2,\ldots,R$
\end{theorem}

The proof of Theorem \ref{thm-general} can be easily carried out by repeating the reasoning presented in Sections \ref{ilp-proof} and \ref{ila-proof}. Note that by putting $R=2$ we can retrieve results from Theorem \ref{thm-main} and for $R=1$ both cases (ILP and ILA) are equivalent and the main result from \cite{bgp} follows.

Another natural generalisation of this result would be redefining the trigger of CoCoCat bond. Let $r_0 = 1,2,\dots,R$ and assume that the CoCoCat bond mechanism is activated when exactly $r_0$ out of $R$ thresholds are exceeded. In pricing, this definition would involve the distribution of $r_0$th order statistic for $R$ random variables. Although it is mathematically doable, the authors are aware that such modification would not produce any neat formulae. The details are therefore left for the reader.

\section{Numerical illustration}
\label{sec:num}

In this section, we investigate the behaviour of the price of multi-region CoCoCat bonds through numerical experiments. We analyse the price as a function of the initial capitals $D_1$ and $D_2$, and the value of the parameter $\nu$ of the conversion price $K_P = S_0^{\nu}.$

The parameters of the models were calibrated using data from the Property Claims Services (PCS). The chosen sample consisted of losses caused by 44 wind, thunderstorm, and winter storm events that occurred at the same time in Oklahoma and Texas between 1985 and 2011. Loss amounts were adjusted using the consumer price index. 

In most of the events, the losses that occurred in Texas were higher than in Oklahoma. However, Oklahoma was not exempt from extreme losses and was heavily affected by tornadoes in May 1990 and May 2010, see Figure \ref{fig:ok_tx_losses_stem_boxplot}(a). The highest losses in Texas were also caused by tornadoes: in April 1994 and April 2011, see Figure \ref{fig:ok_tx_losses_stem_boxplot}(b). The box plots of loss amounts  presented in Figure \ref{fig:ok_tx_losses_stem_boxplot}(c) clearly show skewness in both samples and heavy tails of the losses.

\begin{figure}
    \centering
\includegraphics[width=0.9\linewidth]{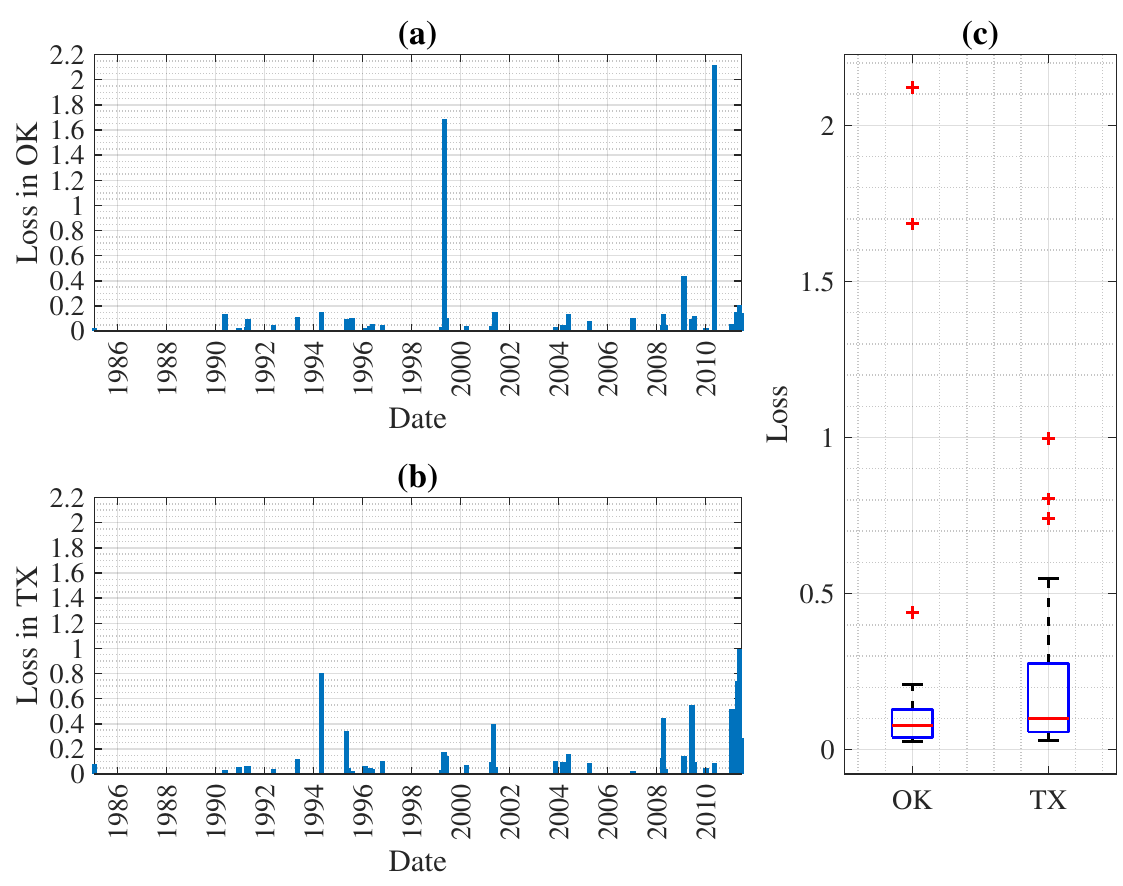}
    \caption{Adjusted losses occurred in (a) Oklahoma and (b) Texas, 
    in billion USD, together with (c) their box plots.}
    \label{fig:ok_tx_losses_stem_boxplot}
\end{figure}

Among log-normal, Pareto, gamma, Weibull, inverse Gaussian and generalised extreme value distributions, the former gave the smallest values of the Kolmogorov-Smirnov (KS), Cramer von Mises (CvM) and Anderson and Darling (AD) test statistics (for the description of the tests we refer to \cite{burjanwer11}). The log-normal distribution is given by the density:
\begin{equation}
    f(x) = \frac{1}{x\sigma\sqrt{2\pi}} \exp
    {\left( - \frac{\left(\ln(x)-\mu \right)^2}{2\sigma^2}\right)}, \quad x>0,\;\mu \in \mathbb{R},\;\sigma>0.
\end{equation}
In Table \ref{tab:logn_params} parameters of the distribution fitted to the losses in each state and the total loss from both states are presented.

\begin{table}[t]
\caption{Parameters of log-normal distribution fitted to data sample used in models with independent loss amounts (ILA) and proportional loss amounts (PLA).}
\label{tab:logn_params}
\centering
\begin{tabular}{ccc} \noalign{\smallskip}\hline\noalign{\smallskip}
Model & State      & Parameters of log-normal distribution\\ \noalign{\smallskip}\hline\noalign{\smallskip}
\multirow{2}{*}{ILA}        &  OK & $\mu=-4.564$, $\sigma =1.813$\\ \noalign{\smallskip}\cline{2-3}\noalign{\smallskip}
 & TX  & $\mu=-2.439$, $\sigma =1.183$    \\ \noalign{\smallskip}\hline\noalign{\smallskip}
 PLA  & total loss & $\mu=-1.477$, $\sigma = 0.902$ \\\hline
\end{tabular}
\end{table}

The timing of catastrophes was modelled by a homogeneous Poisson process with the annual intensity parameter $\lambda=1.4$. The value was estimated using the least squares method by comparing the mean value of the Poisson process $\mathbb{E}N(t) = \Lambda t$ with the aggregate number of events in our data set. Considering a small sample size, the homogeneous Poisson process gave a reasonable fit to the data with the following errors: MSE = 10.13, MAE = 2.53, MAPE = 30.6\%.

For illustration purposes, we consider a 5-year CoCoCat bond, which pays quarterly coupons at an annual rate $10\%$ plus a risk-free rate (like LIBOR). If the bond is triggered, $\zeta = 10\%$ of the nominal is converted into the sponsor's equity at price $S_{\tau}^{\nu}$, where $S_{\tau}$ is the price of the share at the trigger moment $\tau$ and the parameter $\nu \in (0,1).$ For simplicity, we assume that the nominal is $Z=1$.

The Longstaff's model with $r_0=0.02; \theta_r=0.2, \sigma_r = 0.03$ was used to model the interest rate process. The parameters of the stock price process were: $S_0=10, \sigma_s=0.2$. The correlation coefficient of the share and interest rate processes was set to $\rho=-0.5.$ 

\subsection{Results for ILA models}

\begin{figure}[ht]
    \centering
\includegraphics[width=0.9\linewidth]{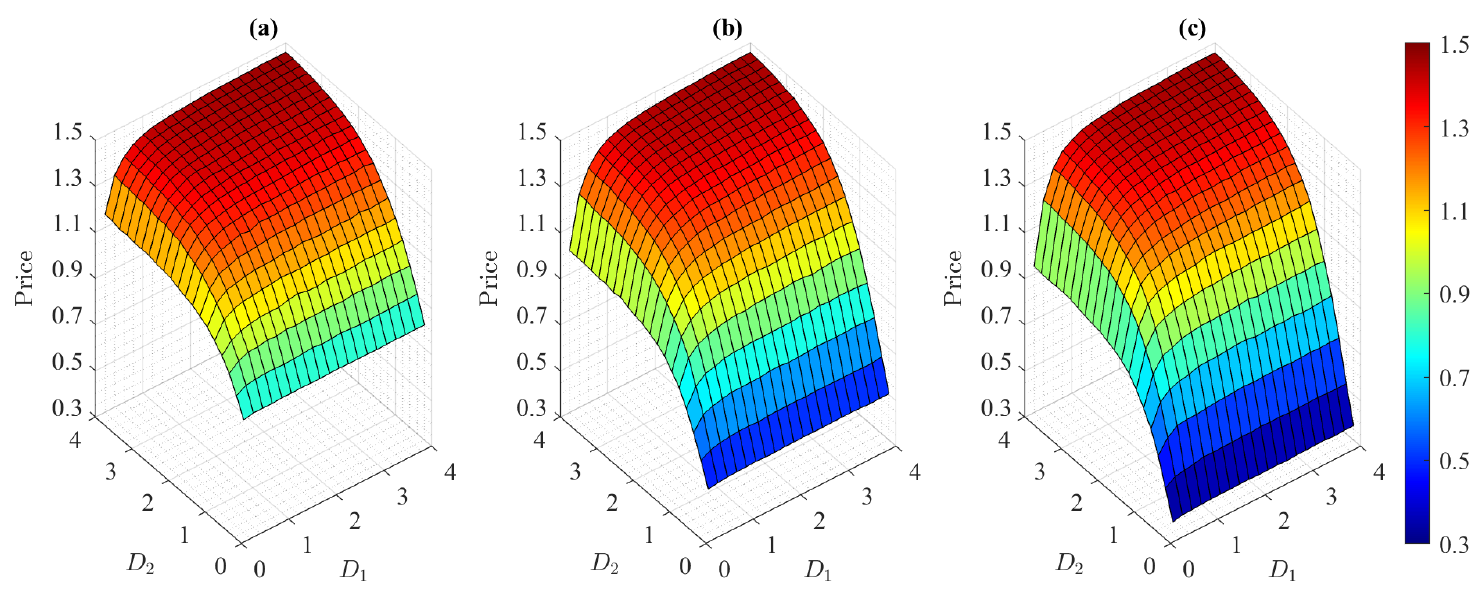}
    \caption{Prices obtained for the ILA model for (a) $\nu=0.2$, (b) $\nu=0.5$, (c) $\nu=0.8$, with respect to $D_1$ and $D_2$ (in billion USD).}    \label{fig:oktx_ILA_prices}
\includegraphics[width=0.9\linewidth]{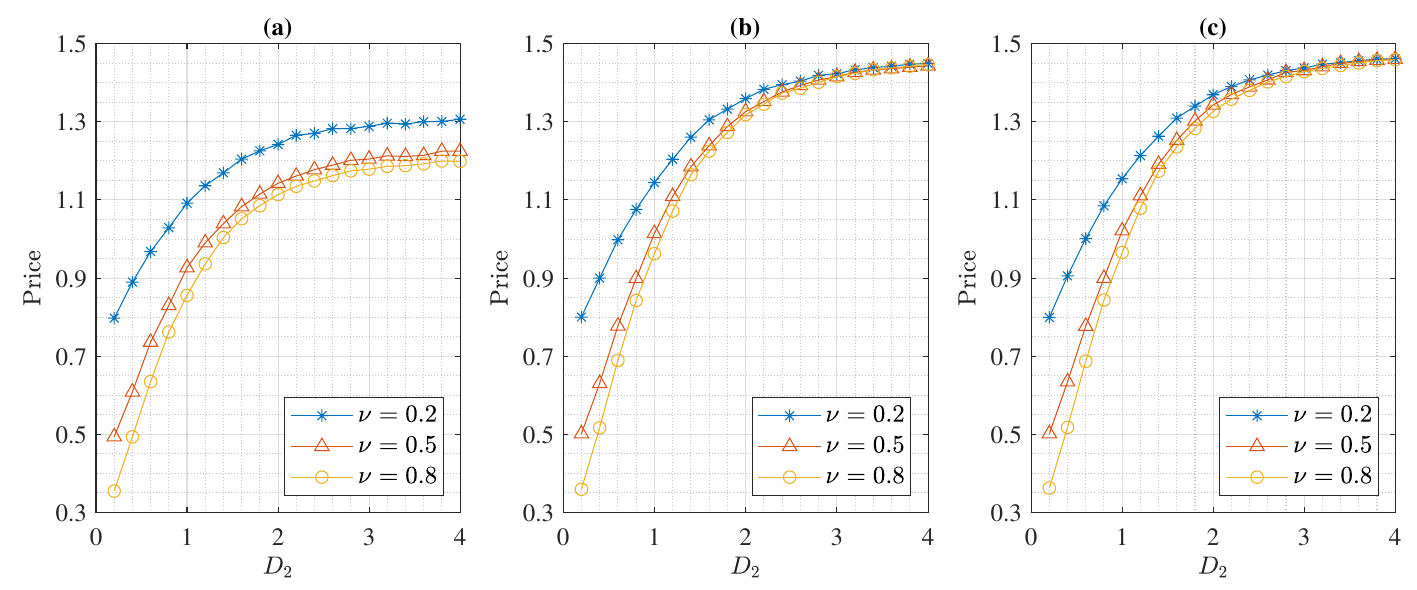}
    \caption{Prices obtained for the ILA model for different values of $\nu$ for (a) $D_1=0.4$, (b) $D_1=2$, (c) $D_1=4$ billion USD, with respect to $D_2$.}
\label{fig:oktx_ila_prices_compare}
\end{figure}

We begin with the analysis of the ILA model, where losses from each state are described as independent random variables. The impact of losses on stock prices of the bond's issuer is described by the parameters.
\begin{equation}
    \alpha = \frac{\delta}{\mathbb{E}[X_k^1]},\quad \beta = \frac{\delta}{\mathbb{E}[X_k^2]},
\end{equation}
where $\delta=0.02$, as in \cite{bgp}. 

In Figure \ref{fig:oktx_ILA_prices}, the price of the CoCoCat bond is presented as a function of $D_1$ and $D_2$ for the three chosen values of the parameter $\nu$. We can observe that as thresholds increase, so does the price of the bond. The higher $D_1$ and $D_2$, the lower the probability of losses exceeding given thresholds, resulting in the payment of all coupons and the nominal. We also observe that as $\nu$ increases, the price decreases, which is visible for small values of $D_2$. 
When the bond is triggered, $\zeta Z$ is converted into the equity at conversion price $S_{\tau}^{\nu}$. For higher values of $\nu$, the owner of the bond gets fewer units of the share since the conversion cost is higher. Consequently, the total value of the shares obtained at the moment of conversion is smaller, explaining the decrease in the price of the CoCoCat bond.
In Figure \ref{fig:oktx_ila_prices_compare} we can clearly see that the impact of the parameter $\nu$ is weakening as the thresholds $D_1$ and $D_2$ increase, since the probability of conversion of the nominal to equity decreases. 

\subsection{Results for PLA models} Next, we move on to models with proportional loss amounts. 
The total loss is divided according to the proportion $P$, that is, for a common loss $X_k$, the value of $P X_K$ is assigned to Oklahoma and $(1-P)X_k$ to Texas. 

For models with proportional loss amounts, the impact of losses on stock prices of the bond's issuer is described by the parameters
\begin{equation}
\alpha = \frac{\delta}{\mathbb{E}[PX_k]}, \quad \beta = \frac{\delta}{\mathbb{E}[X(1-P)]},
\end{equation}
where $\delta = 0.02$ as before.

The average proportion for historical data was found to be approximately equal to $0.38$ and the box plot of the proportions between losses is presented in Figure \ref{fig:ok_tx_proportion}(a). In the model with constant PLA (cPLA), we assume that the proportion $P$ is equal to $0.38$ with probability 1. We can interpret it as OK taking 38\% of the loss and TX taking 62\% of the total loss. 

The price of the CoCoCat bond as a function of thresholds is presented in Figure \ref{fig:oktx_constP_prices}. Similarly as for the ILA model, the higher $D_1$ and $D_2$ are, the higher the price. We can observe that the value of the price quickly converges to the maximum value. Again, an increase in $\nu$ causes a decrease in the price which is especially visible for smaller values of $D_1$ and $D_2$, see Figure \ref{fig:oktx_constP_prices_compare}.

\begin{figure}[ht]
    \centering
\includegraphics[width=0.9\linewidth]{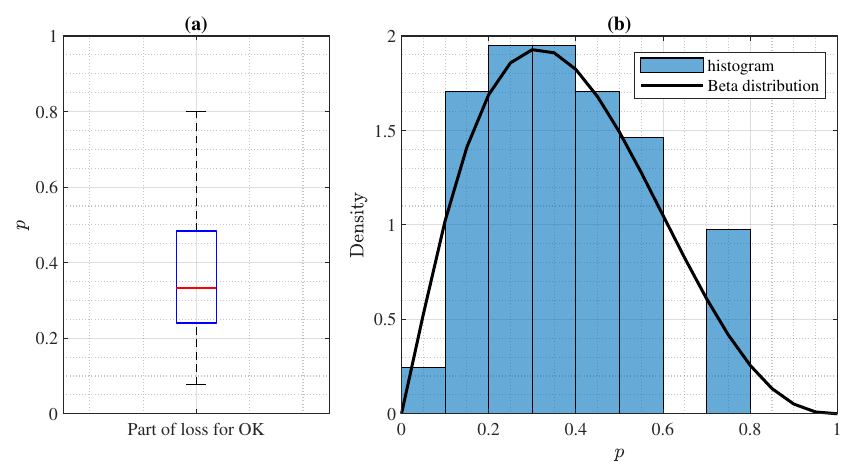}
    \caption{(a) Box plots for proportion between losses. (b) Normalized histogram of proportions compared with the density of fitted beta distribution.}
\label{fig:ok_tx_proportion}
\end{figure}

In the random PLA (rPLA) model, we describe the proportion between losses by the beta distribution with density:
\begin{equation}
    f(x) = \frac{x^{\alpha-1}\left(1-x\right)^{\beta -1}}{B\left(\alpha,\beta\right)},\quad x\in [0,1],\;\alpha>0,\; \beta>0,
\end{equation}
where $B(x)$ is the beta function. Using the maximum likelihood method, the parameters $\alpha=2.1531$, $\beta=3.5135$ were fitted. The expected value of this variable is equal to $0.38$, so it is the same as the constant proportion. The distribution fits the data well, see Figure \ref{fig:ok_tx_proportion}(b). Moreover, KS, CvM and AD tests did not reject the beta distribution hypothesis. 

\begin{figure}[ht]
\includegraphics[width=0.9\linewidth]{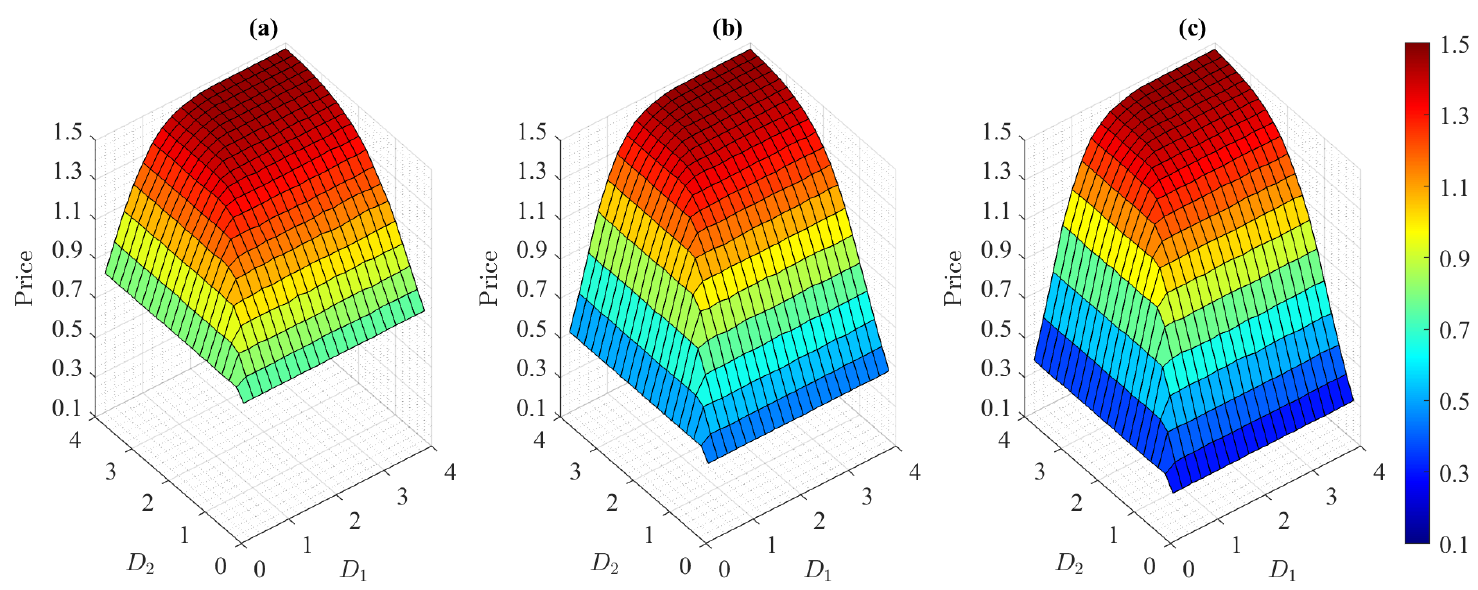}
    \caption{Prices obtained for the cPLA model for (a) $\nu=0.2$, (b) $\nu=0.5$, (a) $\nu=0.8$, with respect to $D_1$ and $D_2$ (in billion USD).}
\label{fig:oktx_constP_prices}    \includegraphics[width=0.8\linewidth]{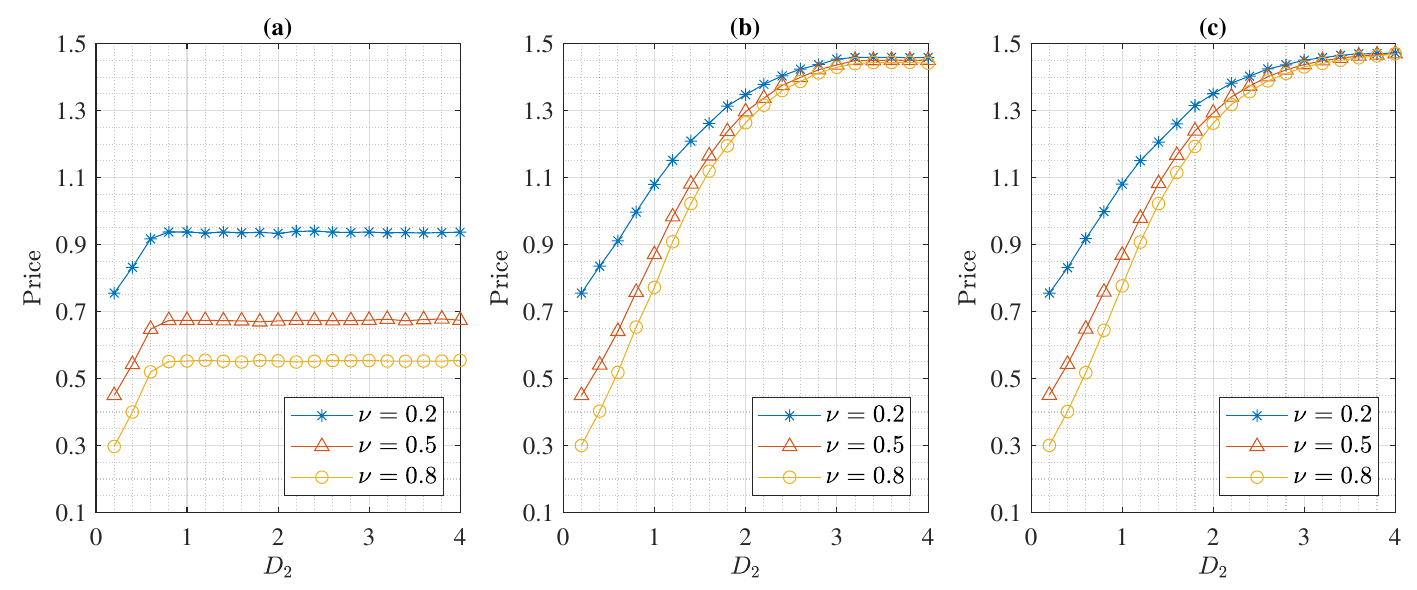}
    \caption{Prices obtained  for the cPLA model for different values of $\nu$ for (a) $D_1=0.4$, (b) $D_1=2$, (c) $D_1=4$ billion USD, with respect to $D_2$.}
\label{fig:oktx_constP_prices_compare}
\end{figure}

The prices of the CoCoCat bond for the rPLA model are presented in Figure \ref{fig:oktx_randomP_prices}. The prices behave similarly to the results for constant proportion, see Figure \ref{fig:oktx_constP_prices}, but the increase in the price as the thresholds grow is slower. The parameter $\nu$ also impacts the price, see Figure \ref{fig:oktx_randomP_prices}, especially for lower thresholds.

\subsection{Comparison of the models}

We take thresholds $D_1$ and $D_2$ as quantiles of order $q$ of the total loss. The prices naturally increase as we take higher quantiles of the distribution (since $D_1$ and $D_2$ are getting higher). We can also clearly observe the impact of $\nu$ on the prices -- the higher $\nu$, the lower the price.
 The effect is the most visible for smaller values of thresholds $D_1$ and $D_2$, when the possibility of conversion is much higher.




\begin{figure}
    \centering
    \includegraphics[width=0.9\linewidth]{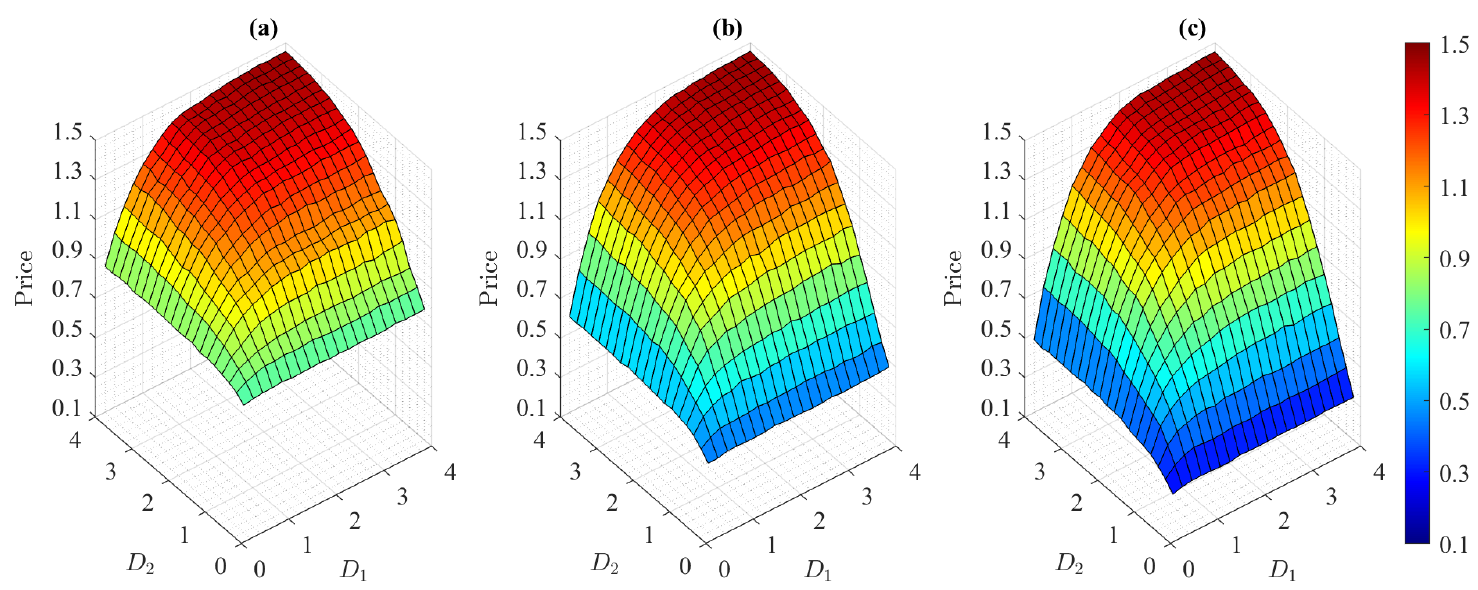}
    \caption{Prices obtained for the rPLA model for (a) $\nu=0.2$, (b) $\nu=0.5$, (a) $\nu=0.8$, with respect to $D_1$ and $D_2$ (in billion USD).}
    \label{fig:oktx_randomP_prices}
    \includegraphics[width=0.9\linewidth]{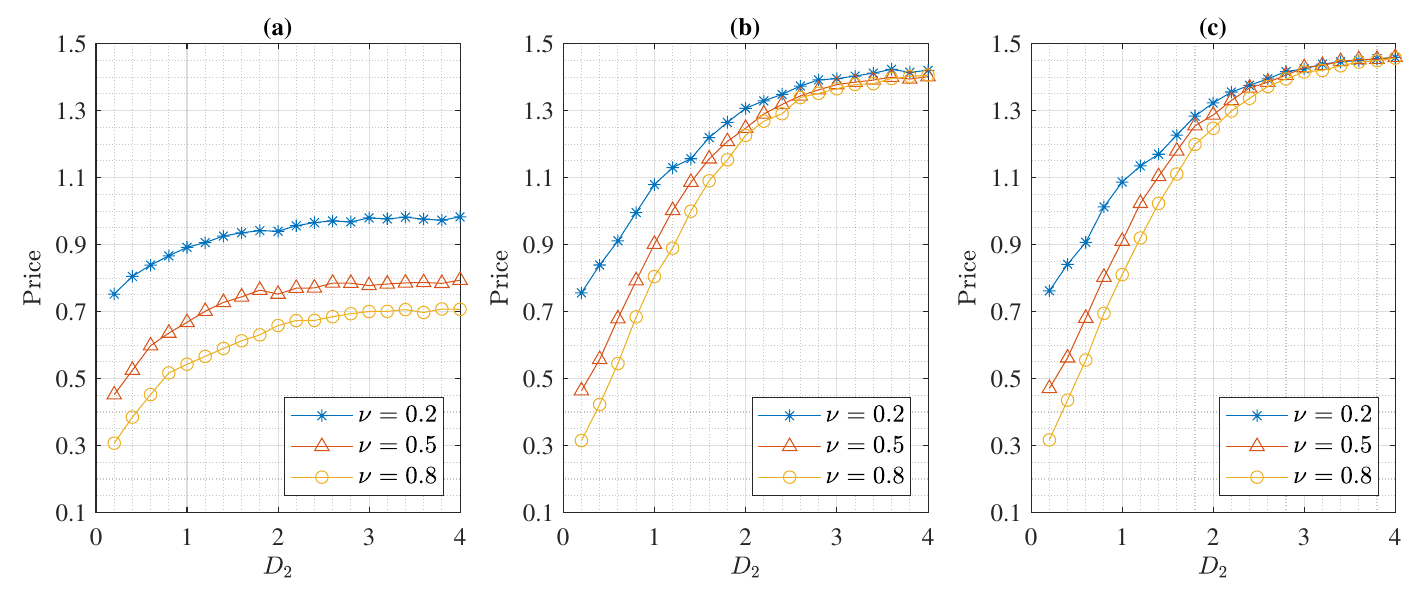}
    \caption{Prices obtained for the rPLA model for different values of $\nu$ for (a) $D_1=0.4$, (b) $D_1=2$, (c) $D_1=4$ billion USD, with respect to $D_2$.}
    \label{fig:oktx_randomP_prices_compare}
    
\end{figure}


\begin{figure}
    \centering
    \includegraphics[width=0.9\linewidth]{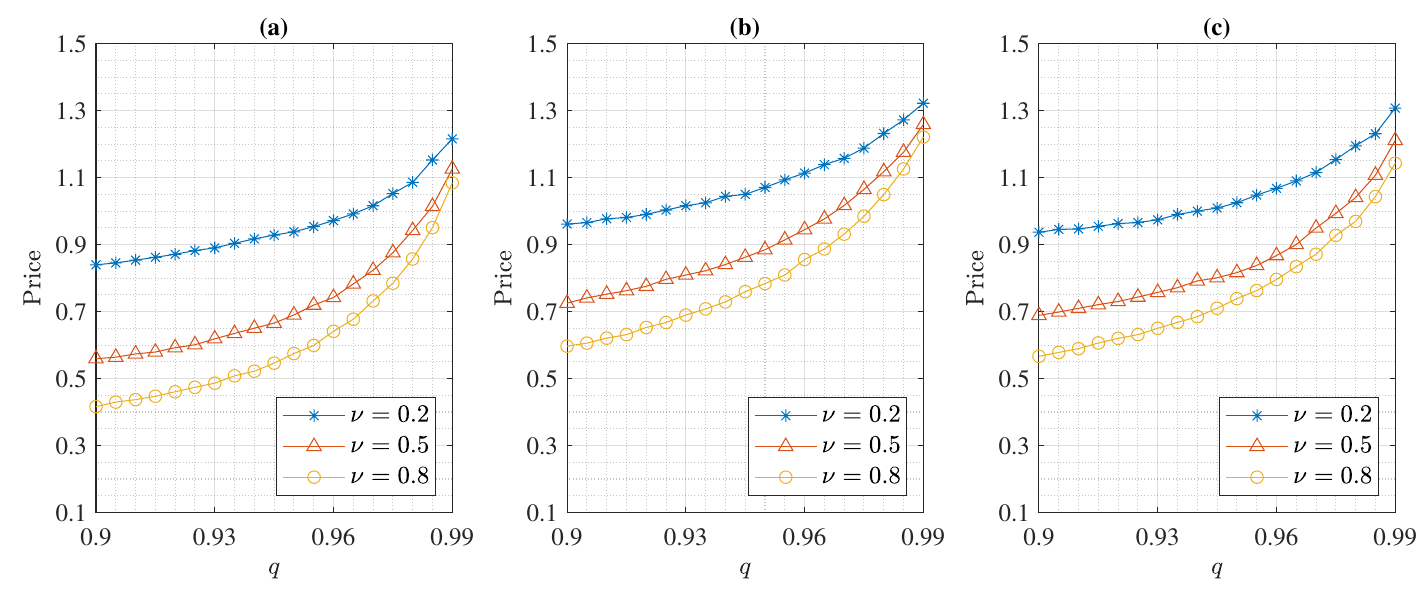}
    \caption{Prices obtained for the (a) ILA, (b) cPLA, (c) rPLA models, for $D_1$ and $D_2$ being $q$-quantiles of respective loss distributions, for different $\nu$, with respect to $q$.}
    \label{fig:oktx_quantileD}
\end{figure}


\section{Conclusions}
\label{sec:con}

Insurance-linked securities, in particular CAT bonds, is a dynamic and evolving market where multi-peril and multi-region structures play an increasingly significant role. These complex instruments represent a sophisticated convergence of insurance risk management and capital markets finance, driven by sponsors' needs for comprehensive coverage and large-scale capacity, and investors' search for yield and diversification.


A CoCoCat, introduced in \cite{bgp} can help stabilise their issuers' balance sheets in times of distress; particularly in times of extreme natural catastrophes potentially spurring on large non-independent insurance-related losses.
In this paper, we addressed a novel insurance-linked instrument which is in the form of a contingent convertible bond (CoCoCat bond) with a trigger linked to the occurrence of predefined natural catastrophes in different regions. To this end, we constructed a pricing model that incorporated different dependence scenarios between regions.


First, we introduced a general 2D model. This comprehensive model is designed to capture both financial market risk and catastrophe risk variables. The financial market component is modelled using classical Black-Scholes dynamics, incorporating a stochastic interest-rate process, for which Longstaff's model is specifically chosen. The catastrophe risk component was addressed by modelling a 2D aggregate loss process, which can be represented by various compound Poisson processes. 

Then, we proposed three specialised cases of modelling the aggregate loss processes, which are crucial for the subsequent pricing formulae, namely
\begin{enumerate}
    \item ILP: This model assumes that losses originate from two distinct regions and that these losses occur with different frequencies. Crucially, under ILP, the variables representing the number of losses and the loss amounts in the two regions are considered pairwise independent, leading to independent aggregate loss processes for each region.
    \item ILA: In this scenario, it is assumed that losses in the two different regions occur simultaneously, meaning that they share the same loss frequency. However, the amounts of these losses ($X_k^1$ and $X_k^2$) are independent of each other.
    \item PLA: This model also assumes that losses occur at the same time. The distinguishing feature here is that the loss amounts themselves are split proportionally between the two regions. Two sub-types were considered:
    \begin{itemize}
        \item cPLA: The proportion of loss allocated to each region is fixed and deterministic.
        \item rPLA: The proportion of loss is a random variable, independent of the loss amounts and their frequency. For this case, we considered the beta distribution.
    \end{itemize}
\end{enumerate}

Next, for all considered cases, we derived the risk-neutral pricing formulas using change-of-measure techniques.  We were able to find intuitive and analytical expressions for the prices. An exponential
change of measure allowed us to separately deal with financial markets as well as catastrophe-risk variables, and a Girsanov-like transformation allowed us to synthetically remove a Brownian motion from the expectation containing two correlated Brownian motions.

 We arrived at an analytical expression for the conversion feature (and hence the price) which only required simulation of the loss process in order to empirically estimate the distribution of the time-of-trigger of the equity conversion feature. We note that Monte Carlo simulation could be used to estimate the value of the conversion feature of the CoCoCat directly. However, our simplification to an analytical formula has more in its favour, since only one process had to be simulated. 

We also fitted the model to the natural catastrophe data provided by Property Claim Services and investigated the influence of the dependence on the bond prices. We presented numerical experiments as a first foray into the price behaviour of the CoCoCat. Gaining an understanding of the IL CoCoCat price behaviour, for varying parameters, is crucial in the design stage of such an instrument. 

The prices we obtained in our analyses were in accordance with intuition: the higher the threshold level of the IL CoCoCat, the greater the price. 
We also found that the conversion fraction significantly impacts the value of the conversion feature; namely, the lower the conversion factor, the higher the price. Finally, we discovered that the prices under different dependence scenarios vary. For the choices of parameters considered, the independence scenario yielded the lowest prices. This justifies the need to choose the right approach. Finally, we note that other dependence scenarios can also be analysed, but a key advantage of the introduced models is their ability to produce simple pricing formulas with parameters that can be readily estimated from loss data. 

\section*{CRediT authorship contribution statement}

\textbf{Jacek Wszoła:} Writing – original draft, Formal analysis, Methodology, Validation, Proofs. \textbf{Krzysztof Burnecki:} Writing – review \& editing, Conceptualization, Methodology, Validation, Supervision. \textbf{Marek Teuerle:} Writing – review \& editing, Conceptualization, Methodology, Validation, Supervision. \textbf{Martyna Zdeb:} Writing – Original Draft, Software, Visualization.

\bibliographystyle{abbrv}
\bibliography{main}

\end{document}